\documentclass[runningheads,a4paper,envcountsame]{llncs} 

\usepackage{amsmath}
\usepackage{amsfonts}
\usepackage{cases}
\usepackage{mathtools}
\usepackage{algpseudocode}
\usepackage{algorithm}
\usepackage{arydshln}
\usepackage{lscape}
\usepackage{booktabs}
\usepackage{temporal}
\usepackage{todonotes}
\usepackage{color}
\usepackage{longtable}
\usepackage{tabu}
\usepackage{tabularx} 
\usepackage{multirow}
\usepackage{cases}
\usepackage{wrapfig}
\usepackage{caption}
\usepackage{subfig}
\usepackage{tikz}
\usetikzlibrary{arrows,automata}
\tikzset{initial text={}}
\tikzset{initial distance=3mm}
\tikzset{every picture/.style=semithick} 
\tikzset{>=stealth'} 
\tikzset{-} 

\algnewcommand\algorithmicswich{\textbf{switch}}
\algnewcommand\algorithmiccase{\textbf{case}}
\algnewcommand\algorithmicassert{\texttt{assert}}
\algnewcommand\Assert[1]{\State \algorithmicassert(#1)}%
\algdef{SE}[SWITCH]{Switch}{EndSwitch}[1]{\algorithmicswitch\ #1\ \algorithmicdo}{\algorithmicend\ \algorithmicswitch}%
\algdef{SE}[CASE]{Case}{EndCase}[1]{\algorithmiccase\ #1}{\algorithmicend\ \algorithmiccase}%

\algtext*{EndSwitch}%
\algtext*{EndCase}%

\newcommand{\impossible}{-}
\newcommand{\threeval}{\mu_\pi}
\newcommand{\countsem}{d_\pi}
\DeclareMathOperator{\swap}{\sim}

\newcommand{\inc}{\oplus 1}
\newcommand{\minmax}{\sqcup}
\newcommand{\maxmin}{\sqcap}

\newcommand{\setNpos}{\mathbb{N}_{>0}}

\newcommand{\predictive}{e_\pi}
\newcommand{\fsem}{r_\pi}
\newcommand{\ptrue}{\top\!_P}
\newcommand{\pfalse}{\bot_P}
\newcommand{\inconclusive}{{?}}
\newcommand{\prediction}{\text{pred}}

\begin{document}

\title{A Counting Semantics for Monitoring  
LTL Specifications over Finite Traces}
\author{E. Bartocci\inst{1}, R. Bloem\inst{2}, D. Nickovic\inst{3} and F. Roeck\inst{2}}

\institute{
$^1$ Vienna University of Technology, Vienna, Austria \\
$^2$ Graz University of Technology, Graz,  Austria \\
$^3$ Austrian Institute of Technology GmbH, Vienna, Austria
} 

\maketitle

\begin{abstract}
We consider the problem of monitoring a Linear Time Logic (LTL)
specification that is defined on infinite paths, over finite traces.  
For example, we may need to draw a verdict on whether
the system satisfies or violates the property ``p holds
infinitely often.''  The problem is that there is
always a continuation of a finite trace that satisfies the
property and a different continuation that violates it.

We propose a two-step approach to address this problem. First, we
introduce a counting semantics that computes the number of steps
to witness the satisfaction or violation of a formula for each
position in the trace.  Second, we use this information to make a
prediction on inconclusive suffixes.  In particular, we consider a
\emph{good} suffix to be one that is shorter than the longest witness
for a satisfaction, and a \emph{bad} suffix to be shorter than or
equal to the longest witness for a violation. Based on this
assumption, we provide a verdict assessing whether a continuation of
the execution on the same system will presumably satisfy or violate
the property.
\end{abstract}


\section{Introduction}
\label{sec:introduction}



Alice is a verification engineer and she is presented with a new exciting and complex design.  
The requirements document coming with the design already incorporates functional requirements 
formalized in Linear Temporal Logic (LTL)~\cite{DBLP:conf/focs/Pnueli77}. 
The design contains features that are very challenging for exhaustive verification and her 
favorite model checking tool does not terminate in reasonable time. 

\paragraph{Runtime Verification.}
Alice decides to tackle this problem using runtime verification (RV),
 a light, yet rigorous verification method. 
RV drops the exhaustiveness of model checking and analyzes individual traces generated by 
the system. Thus, it scales much better to the industrial-size designs.  RV
can be directly applied to the design, and does not require its abstract model. This method enables  
automatic generation of monitors from formalized requirements and thus provides a systematic way to check 
whether the system executions satisfy or violate the specification.  

\paragraph{Motivating Example.}
During her RV activities, Alice comes across the following unbounded response requirement:
$$
\psi \equiv \textsf{G(request $\rightarrow$ F grant)}
$$
\noindent This formula says that every request coming from the environment must be granted by the design in some finite (but unbounded) future. 
Alice realizes that she is trying to check a \emph{liveness} property over a set of \emph{finite} traces. 
She looks closer at the executions and identifies the two interesting
examples trace $\tau_{1}$ and trace $\tau_{2}$, depicted in Table~\ref{tab:response}.

\begin{wraptable}{r}{0.4\textwidth}
\captionsetup{width=0.4\textwidth}
\vspace{-10mm}
            \centering
\caption{Unbounded response property example.}
\begin{tabular}{c|r|ccccccc}
\toprule
\textsf{trace}	& \textsf{time}		 &  1	& 2	& 3	& 4	& 5	& 6	& 7	 \\
\midrule
$\tau_{1}$   & \textsf{request} 	& $\top$ & $-$ & $-$ & $\top$ & $-$ & $-$ & $-$  \\
	          & \textsf{grant}	&  $-$ & $-$ & $\top$ & $-$ & $-$ & $-$ & $-$  \\

\midrule

$\tau_{2}$    & \textsf{request} 	& $\top$ & $-$ & $-$ & $\top$ & $-$ & $-$ & $\top$  \\
                      & \textsf{grant}	&  $-$ & $-$ & $\top$ & $-$ & $-$ & $\top$ & $-$\\
\bottomrule
\end{tabular}
\label{tab:response}

\vspace{-3mm}
 \end{wraptable}

The runtime verification tool reports that both $\tau_{1}$ and $\tau_{2}$ presumably violate the unbounded response property. This verdict is 
against Alice's intuition. The evaluation of trace $\tau_{1}$ seems right to her -- the request at Cycle $1$ is followed by a grant at 
Cycle $3$, however the request at Cycle $4$ is never granted during that execution. There are good reasons to suspect a bug in the 
design. Then she looks at $\tau_{2}$ and observes that after every \textsf{request} the \textsf{grant} is 
given exactly after $2$ cycles.  It is true that the last request at Cycle $7$ is not followed by a grant, but this seems to happen 
because the execution ends at that cycle -- the past trace observations give reason to think that this request would be followed by a 
grant in cycle $9$ if the execution was continued. Thus, Alice is not satisfied by the second verdict.

Alice looks closer at the way that the LTL property is evaluated over finite traces. She finds out that temporal operators are given {\em strength} -- {\em eventually} and 
{\em until} are declared as {\em strong} operators, while {\em always} and {\em weak until} are defined to be {\em weak}~\cite{DBLP:conf/cav/EisnerFHLMC03}.  
A strong temporal operator requires all outstanding obligations to be met before the end of the trace. In contrast, a weak temporal operator 
must not witness any outstanding obligation violation before the end of the trace. Under this interpretation, both $\tau_{1}$ and $\tau_{2}$ 
violate the unbounded response property.

Alice explores another popular approach to evaluate future temporal properties over finite traces -- the $3$-valued semantics for LTL~\cite{DBLP:conf/fsttcs/BauerLS06}. In this setting, the Boolean 
set of verdicts is extended with a third \textsf{unknown} (or \textsf{maybe}) value. A finite trace satisfies (violates) the  
$3$-valued LTL formula if and only if all the infinite extensions of the trace satisfy (violate) the same LTL formula under its classical interpretation. In all other 
cases, we say that the satisfaction of the formula by the trace is \textsf{unknown}. 
Alice applies the $3$-valued interpretation of LTL on the traces $\tau_{1}$ and $\tau_{2}$ to evaluate the unbounded response property. 
In both situations, she ends up with the \textsf{unknown} verdict. Once again, this is not 
what she expects and it does not meet her intuition about the satisfaction of the formula by the observed traces.

Alice desires a semantics that evaluates LTL properties on finite traces by taking previous observations into account.

\paragraph{Contributions.}
In this paper, we study the problem of LTL evaluation over finite traces encountered by Alice and propose a solution. 
We introduce a new counting semantics for LTL that takes into 
account the intuition illustrated by the example from Table~\ref{tab:response}. This semantics computes for every position of a trace two values -- 
the distances to the nearest satisfaction and violation of 
the co-safety, respectively safety, part of the specification. We use this quantitative information 
to make \emph{predictions} about the (infinite) suffixes of the finite observations.  We infer from these values the maximum time 
that we expect for a future obligation to be fulfilled. We compare it to the value that we have for an open obligation at the 
end of the trace. If the latter is greater (smaller) than the expected maximum value, we have a good indication of a \emph{presumed violation (satisfaction)} that 
we report to the user.
In particular, our approach will indicate that $\tau_{1}$ is likely to 
violate the specification and should be further inspected. In contrast, it will evaluate that $\tau_{2}$ most likely satisfies the unbounded response property.
\paragraph{Organization of the paper.}  The rest of the paper is organized as follows.  
We discuss the related work in Section 2 and we provide the preliminaries in Section 3. 
In Section 4 we present our new counting semantics for LTL, while in Section 5 
we show how to make \emph{predictions} about the (infinite) suffixes of the finite observations. 
Section 6 shows the application of our approach to some examples. 
Finally in Section 7 we draw our conclusions.



\section{Related Work}
\label{sec:relative}




The finitary interpretation of LTL was first considered in~\cite{DBLP:books/daglib/0077033}, where the authors propose to enrich the logic with the 
{\em weak} next operator that is dual to the (strong) next operator defined on infinite traces. While the strong next requires the existence of a next state, the 
weak next trivially evaluates to true at the end of the trace. 
%
%
%
In~\cite{DBLP:conf/cav/EisnerFHLMC03}, the authors propose a more semantic approach with 
{\em weak} and {\em strong} views for evaluating future obligations at the end of the 
trace. In essence the empty word satisfies (violates) every formula according to the weak (strong) view.
These two approaches result in the violation of the specification $\psi$ by both traces $\tau_{1}$ and $\tau_{2}$.  

The authors in~\cite{DBLP:conf/fsttcs/BauerLS06} propose a $3$-valued finitary LTL interpretation of LTL, in which 
the set $\{ \textsf{true}, \textsf{false} \}$ of verdicts is extended with a third $\textsf{inconclusive}$ verdict. According to the 
$3$-valued LTL, a finite trace satisfies (violates) a specification iff all its infinite extensions satisfy (violate) the same 
property under the classical LTL interpretation. Otherwise, it evaluates to $\textsf{inconclusive}$. The main disadvantage of 
the $3$-valued semantics is the dominance of the $\textsf{inconclusive}$ verdict in the evaluation of many interesting 
LTL formulas. In fact, both $\tau_{1}$ and $\tau_{2}$ from Table~\ref{tab:response} evaluate to $\textsf{inconclusive}$ 
against the unbounded response specification $\psi$.


In~\cite{DBLP:conf/rv/0002LS07}, the authors combine the weak and strong operators with the $3$-valued semantics 
to refine the $\textsf{inconclusive}$ with $\{ \textsf{presumably true}, \textsf{presumably false} \}$.  The strength of  the remaining future 
obligation dictates the presumable verdict. The authors in~\cite{DBLP:conf/lpar/MorgensternGS12} propose a 
finitary semantics for each of the LTL (safety, liveness, persistence and 
recurrence) hierarchy classes that asymptotically converges to the infinite traces semantics of the logic.
In these two works, the specification $\psi$ also evaluates to 
the same verdict for both the traces $\tau_{1}$ and $\tau_{2}$.

To summarize, none of the related work handles the unbounded response example from Table~\ref{tab:response} in a satisfactory 
manner. This is due to the fact that these approaches decide about the verdict based on the specification and its remaining future 
obligations at the end of the trace. In contrast, we propose an approach in which the past observations within the trace are used to 
predict the future and derive the appropriate verdict. In particular, the application of our semantics for the evaluation 
of $\psi$ over $\tau_{1}$ and $\tau_{2}$ results in $\textsf{presumably true}$ and $\textsf{presumably false}$ verdicts.


In~\cite{DBLP:conf/nfm/ZhangLD12}, the authors propose another predictive semantics for LTL. In essence, this work assumes that 
at every point in time the monitor is able to precisely predict a segment of the trace that it has not observed yet and produce its outcome 
accordingly. In order to ensure such predictive power, this approach requires a white-box setting in which instrumentation and some 
form of static analysis of the systems are needed in order to foresee in advance the upcoming observations. This is in contrast to our 
work, in which the monitor remains a passive participant and predicts its verdict only based on the past observations. 

In a different research thread~\cite{DBLP:conf/ictac/ViswanathanK04}, the authors introduce the notion of 
{\em monitorable} specifications that can be positively or negatively determined by a finite trace. The 
monitorability of LTL is further studied in~\cite{PnueliZ06,DBLP:journals/tosem/BauerLS11}. This classification 
of specifications is orthogonal to our work. We focus on providing a sensible evaluation to all LTL properties, 
including the non-monitorable ones (e.g., $\always \! \eventually p$).


We also mention the recent work on statistical model checking for LTL~\cite{DBLP:conf/tacas/DacaHKP16}. In this work, the 
authors assume a gray-box setting, where the system-under-test (SUT) is a Markov chain with the known minimum transition probability. 
This is in contrast to our work, in which we passively observe existing finite traces generated by the SUT, i.e., we have a blackbox setting.

In~\cite{AlmagorBK14}, the authors 
propose extending LTL with a discounting operator and study the properties of the augmented logic. The LTL specification 
formalism is extended with path-accumulation assertions in~\cite{Boker:2014}. These LTL extensions are 
motivated by the need for a more quantitative and refined analysis of the systems. In our work, the motivation for the 
counting semantics is quite different. We use the quantitative information that we collect during the execution of the 
trace to predict the future behavior of the system and thus improve the quality of the monitoring verdict.


\section{Preliminaries}
\label{sec:prelim}

We first introduce {\em traces} and Linear Temporal Logic (LTL) that we interpret over $3$-valued semantics.

\begin{definition}[Trace]
Let $P$ a finite set of {\em propositions} and let $\Pi=2^P$. 
A (finite or infinite) {\em trace} $\pi$ is a sequence $\pi_{1}, \pi_{2}, \ldots \in \Pi^* \cup \Pi^\omega$ . 
We denote by $|\pi| \in \mathbb{N} \cup \{ \infty \}$ the {\em length} of $\pi$.
We denote by $\pi \cdot \pi'$ the concatenation of $\pi \in \Pi^*$ and $\pi' \in \Pi^* \cup \Pi^\omega$.
\end{definition}

\begin{definition}[Linear Temporal Logic]
In this paper, we consider linear temporal logic (LTL) and we define its syntax by the grammar:
$$
\phi := p~|~\neg \phi~|~\phi_{1} \vee \phi_{2}~|~\nextt \phi~|~\phi_{1} \until \phi_{2},
$$
\noindent where $p \in P$. We denote by $\Phi$ the set of all LTL formulas.

\end{definition}

%
%
%
%
%
%
%
%

 From the basic definition we can derive other standard 
 Boolean and temporal operators  as follows: 
 \vspace{-2ex} $$\top = p \lor \neg p, \mbox{    } \bot=\neg \top, \mbox{    }  \phi \wedge \psi = \neg (\neg \phi \vee \neg \psi), \mbox{    }  \eventually \phi = \top \until \phi,  \mbox{    } \always \phi = \neg \eventually \neg \phi $$
 
Let $\pi \in \Pi^{\omega}$ be an infinite trace and $\phi$ an LTL formula. 
The satisfaction relation $(\pi, i) \models \phi$ is defined inductively as 
follows
$$
\begin{array}{lcl}
(\pi, i) \models p & \textrm{iff} & p \in \pi_{i}\text{,} \\
(\pi, i) \models \neg \phi & \textrm{iff} & (\pi, i) \not \models \phi \text{,}\\
(\pi, i) \models \phi_{1} \vee \phi_{2} & \textrm{iff} & (\pi, i) \models \phi_{1} \; \textrm{or} \; (\pi, i) \models \phi_{2}  \text{,}\\
(\pi, i) \models \nextt \phi & \textrm{iff} & (\pi, i+1) \models \phi \text{,}\\
(\pi, i) \models \phi_{1} \until \phi_{2} & \textrm{iff} & \exists j \geq i \; \textrm{s.t.}\; (\pi, j) \models \phi_{2} \: \textrm{and} \;
\forall i \leq k < j, (\pi, k) \models \phi_{1} \text{.}\\
\end{array}
$$ 

We now recall the $3$-valued semantics from~\cite{DBLP:conf/fsttcs/BauerLS06}. We denote by $[\pi \models_{3} \phi]$ the evaluation of $\phi$ with respect to the 
trace $\pi \in \Pi^{*}$ that yields a value in $\{ \top, \bot, \inconclusive \}$. 

$$
\begin{array}{lll}
[\pi \models_{3} \phi]
				& =	& 

\begin{split}
  \begin{cases}
    \top & \forall \pi' \in \Pi^{\omega}, \pi \cdot \pi' \models \phi\text{,} \\
    \bot & \forall \pi' \in \Pi^{\omega}, \pi \cdot \pi' \not \models \phi\text{,} \\
    \inconclusive & \text{otherwise}\text{.} \\
    \end{cases}  
\end{split}

\\ 

\end{array}
$$

We now restrict LTL to a fragment without explicit $\top$ and $\bot$ symbols and with the explicit $\eventually$ operator that we add to the 
syntax. We provide an alternative $3$-valued semantics for this fragment, denoted by $\threeval(\phi, i)$ where $i \in \setNpos$ indicates a position in or outside the trace. We 
assume the order $\bot < ? < \top$, and extend the Boolean operations to the $3$-valued domain with the rules 
 $\neg_3 \top = \bot$, $\neg_3 \bot = \top$ and $\neg_3 ? = ?$ and $\phi_1 \lor_3 \phi_2 = max(\phi_1,\phi_2)$. We define the semantics inductively as follows:
 
$$
\begin{array}{lll}
\threeval(p, i)					& =	& 

\begin{split}
  \begin{cases}
    \top & \text{if } i \leq |\pi| \; \textrm{and} \; p \in \pi_{i}\text{,} \\
    \bot & \text{else if } i \leq |\pi| \; \textrm{and} \; p \not \in \pi_{i}\text{,} \\
    \inconclusive & \text{otherwise}\text{,} \\
    \end{cases}  
\end{split}

\\ 

\threeval(\neg \phi, i)				& =	& \neg_3 \threeval(\phi,  i)\text{,} \\
\threeval(\phi_{1} \vee \phi_{2}, i)		& =	& \threeval(\phi_{1},  i) \vee_3 \threeval(\phi_{2},  i)\text{,} \\
\threeval(\nextt \phi,  i)			& =	& \threeval(\phi,  i+1)\text{,} \\
\threeval(\eventually \phi, i)			& =	& 
\begin{split}
  \begin{cases}
    \threeval(\phi, i) \vee_3 \threeval(\nextt \! \eventually \phi, i)  & \text{if } i \leq |\pi|\text{,} \\
    \threeval(\phi,  i)  & \text{if } i > |\pi|\text{,} \\
    \end{cases}  
\end{split} \\
\threeval(\phi_{1} \until \phi_{2}, i)		& =	& 
\begin{split}
  \begin{cases}
    \threeval(\phi_{2},  i) \vee_3 (\threeval(\phi_{1},  i) \wedge_3 \threeval(\nextt (\phi_{1} \until \phi_{2}), i))  & \text{if } i \leq |\pi| \text{,}\\
    \threeval(\phi_{2},  i)  & \text{if } i > |\pi| \text{.}\\
    \end{cases}  
\end{split}\\
\end{array}
$$


\noindent 
We note that the adapted semantics allows evaluating a finite trace in
polynomial time, in contrast to $[\pi \models_{3} \phi]$, which
requires a $\textsc{PSPACE}$-complete algorithm. This improvement in
complexity comes at a price -- the adapted semantics cannot
semantically characterize tautologies and contradiction. We have for
example that $\threeval(p \vee \neg p, 1)$ for the empty word
evaluates to $\inconclusive$, despite the fact that $p \vee \neg p$ is
semantically equivalent to $\top$.  The novel semantics that we
introduce in the following sections make the same tradeoff.

In the following lemma, we relate the two three-valued semantics.

\begin{lemma}
\label{lemma:soundness}
Given an LTL formula and a trace $\pi \in \Pi^{*}$, $|\pi| \neq 0$, we have that
$$
\begin{array}{lll}
\threeval(\phi,1) = \top & \Rightarrow & [\pi \models_{3} \phi] = \top \text{,}\\
\threeval(\phi,1) = \bot & \Rightarrow & [\pi \models_{3} \phi] = \bot \text{.}\\
\end{array} 
$$
\end{lemma}

\begin{proof} These two statements can be proven by induction on the structure of the LTL formula 
(see  Appendix~\ref{proof:lemma:soundness}). $ [\pi \models_{3} \phi] = \; \inconclusive \Rightarrow \threeval(\phi,1) = \; \inconclusive $ is the consequence of the first two.

\end{proof}

\section{Counting Finitary Semantics for LTL}
\label{sec:counting}

In this section, we introduce the counting semantics for LTL. We first provide necessary 
definitions in Section~\ref{sec:def}, we present the new semantics in Section~\ref{sec:sem} 
and finally propose a predictive mapping that transforms the counting semantics into a 
qualitative $5$-valued verdict in Section~\ref{sec:evaluation}.


\subsection{Definitions}
\label{sec:def}

Let $\mathbb{N}_{+} = \mathbb{N}_0 \cup \{ \infty, \impossible \}$ be the set of {\em natural} 
numbers (incl. 0) extended with the two special symbols $\infty$ ($\textsf{infinite}$) and $\impossible$ ($\textsf{impossible}$) such that  
$\forall n \in \mathbb{N}_0$, we define $n < \infty < \impossible$. We define the addition $\oplus$ of two elements $a,b \in \mathbb{N}_{+}$ 
as follows.

\begin{definition}[Operator $\oplus$]
 We define the binary operator $\oplus: \mathbb{N}_{+} \times \mathbb{N}_{+} \rightarrow \mathbb{N}_{+}$ 
s. t. for $a \oplus b$ with $a, b \in  \mathbb{N}_{+}$ we have $ a + b$ if  $a,b \in \mathbb{N}_0$ and $\max\{ a,b \}$ otherwise.
\end{definition}

We denote by $(s,f)$ a pair of two extended numbers $s,f \in \mathbb{N}_{+}$. In Definition~\ref{def:pairs}, we introduce several operations on 
pairs: (1) the {\em swap} between the two values ($\swap$), (2) the increment by $1$ of both values ($\inc$), (3) the {\em minmax} binary operation ($\minmax$)
that gives the pair consisting of the minimum first value and the maximum second value, and (4) the {\em maxmin} binary operation ($\maxmin$) that is 
symmetric to  ($\minmax$).

Definition~\ref{def:counting_finitary} introduces the counting semantics for LTL that for a finite trace 
$\pi$ and LTL formula $\phi$ gives a pair $(s,f) \in \mathbb{N}_{+} \times \mathbb{N}_{+}$. 
We call $s$ and $f$ {\em satisfaction} and {\em violation witness counts}, respectively. 
Intuitively, the $s$ ($f$) value denotes the minimal number of additional steps that is needed to witness the satisfaction (violation) 
of the formula. The value $\infty$ is used to denote that the property can be satisfied (violated) only in an infinite 
number of steps, while $\impossible$ means the property cannot be satisfied (violated) by any continuation of the trace.


\begin{definition}[Operations $\swap$, $\inc$, $\minmax$, $\maxmin$]
\label{def:pairs}
\noindent Given two pairs $(s,f) \in \mathbb{N}_+ \times \mathbb{N}_+$ and $(s',f')\in \mathbb{N}_+ \times \mathbb{N}_+$, we have:

$$
\begin{array}{rcl}
\swap (s,f) & = &  (f,s)\text{,}  \\
(s,f) \inc & = &  (s \inc ,f \inc) \text{,}\\
(s,f)  \minmax (s',f') & = & (\min(s,s'), \max(f,f')) \text{,}\\
(s,f)  \maxmin (s',f')& = & (\max(s,s'), \min(f,f')) \text{.} \\
\end{array}
$$
\end{definition}

\begin{example}
\noindent Given the pairs $(0,0)$, $(\infty,1)$ and
$(7,-)$ we have the following:
$$
\begin{array}{rclrcl}
\swap (0,0) & = &  (0,0)\text{,} \;\;\;\;& \swap (\infty,1) & = &  (1,\infty)\text{,}  \\
(0,0) \inc & = &  (1,1) \text{,} \;\;\;\;& (\infty,1) \inc & = &  (\infty,2) \text{,} \\
(0,0) \minmax (\infty,1) & = &  (0,1)\text{,}\;\; \;\;& (\infty,1) \minmax (7,-) & = &  (7,-)\text{,} \\
(0,0) \maxmin (\infty,1) & = &  (\infty,0)\text{,} \;\;\;\;& (\infty,1) \maxmin (7,-) & = &  (\infty,1)\text{.}  \\
\end{array}
$$
\end{example}

\noindent \textbf{Remark.} Note that $\mathbb{N}_+ \times \mathbb{N}_+$  forms a lattice where $(s,f) \unlhd  (s',f')$ when $s \geq s'$  and  $f \leq f'$ with join $\minmax$ and meet $\maxmin$. Intuitively, larger values are closer to true.

\subsection{Semantics}
\label{sec:sem}

We now present our finitary semantics.

\begin{definition}[Counting finitary semantics] 
\label{def:counting_finitary}
 Let $\pi \in \Pi^{*}$ be a finite trace, $i \in \setNpos$ be a position in or outside the trace and $\phi \in \Phi$
 be an LTL formula.
We define the counting finitary semantics of LTL as the 
function\\ $\countsem~:~\Phi \times \Pi^{*} \times \setNpos \rightarrow 
\mathcal{P}(\mathbb{N}_{+} \times \mathbb{N}_{+})$  such that:

{
$$
\begin{array}{lcl}
\countsem(p,i) & = & \begin{cases}
 (0, -)  & \text{if } i \leq |\pi| \wedge p \in \pi_{i}\text{,} \\
 (-, 0)  & \text{if } i \leq |\pi| \wedge p \not \in \pi_{i}\text{,} \\
 (0,0)  & \text{if } i > |\pi| \text{,} \\
\end{cases} \\
\countsem(\neg \phi,i) & = & \swap \countsem(\phi, i )\text{,} \\
\countsem(\phi_{1} \vee \phi_{2},i) & = & \countsem(\phi_{1}, i) \minmax \countsem(\phi_{2},  i)\text{,} \\
\countsem(\nextt \phi,i) & = & \countsem(\phi,  i + 1) \inc\text{,} \\
\countsem( \phi \until \psi, i) & = & \begin{cases}
                                      \countsem(\psi, i) \minmax   
                                       \Big( \countsem(\phi, i) \maxmin  \countsem(\nextt (\phi \until \psi), i )   \Big)  &  \text{if  } i \leq |\pi|\text{,} \\      
                                      \countsem(\psi, i)  \minmax \Big( \countsem(\phi, i)  \maxmin  (-, \infty) \Big) & \text{if  }  i >  |\pi|\text{,} \mbox{ }\\    
                                   \end{cases}\\
                                   
\countsem( \eventually \phi, i) & = & \begin{cases}
                                     \countsem(\phi,  i)  \minmax  \countsem(\nextt \! \eventually \phi, i )     &  \text{if  } i \leq |\pi|\text{,} \\      
                                     \countsem(\phi,  i) \minmax  (-, \infty)   & \text{if  }  i >  |\pi|\text{.} \mbox{ }\\    
                                   \end{cases}\\
\end{array}
$$}
\end{definition}

We now provide some motivations behind the above definitions. 
\begin{description}
\item[Proposition] A proposition is either evaluated before or after the end of the trace. If it is evaluated before the end of the trace and 
the proposition holds, the satisfaction and violations witness counts are trivially $0$ and $\impossible$, respectively. In the case that the 
proposition does not hold, we have the symmetric witness counts. Finally, we take an optimistic view in case of evaluating a 
proposition after the end of the trace: The trace can be extended to a trace with $i$ steps s.t. either $p$ holds or $p$ does not hold.
\item[Negation] Negating a formula simply swaps the witness counts. If we witness the satisfaction of $\phi$ in $n$ steps, we witness 
the violation of $\neg \phi$ in $n$ steps, and vice versa.
\item[Disjunction] 
We take the shorter satisfaction witness count, because the satisfaction of one subformula is enough to satisfy the property. And we take the longer violation witness count, because both subformulas need to be violated to violate the property.
\item[Next] The next operator naturally increases the witness counts by one step.
\item[Eventually] We use the rewriting rule $\eventually \phi \equiv \phi \vee \nextt \! \eventually \phi$ to define the semantics of the 
eventually operator. When evaluating the formula after the end of the trace, we replace the remaining obligation 
($\nextt \! \eventually \phi$) by $(-,\infty)$.
Thus, $\eventually \phi$ evaluated on the empty word is satisfied by a suffix that satisfies $\phi$, and it is violated 
only by infinite suffixes.
\item[Until] We use the same principle for defining the until semantics that we used for the eventually operator. We use the rewriting rule $\phi \until \psi \equiv \psi \lor (\phi \land \nextt (\phi \until \psi) )$. On the empty word, $\phi \until \psi$ is satisfied (in the shortest way) by a suffix that satisfies $\psi$, and it is violated by a suffix that violates both $\phi$ and $\psi$. 
\end{description}

\begin{example}

We refer to our motivating example from Table~\ref{tab:response} and evaluate the trace $\tau_{2}$ with respect to the 
specification $\psi$. We present the outcome in Table~\ref{tab:response-sem}. We see that every proposition evaluates to 
$(0,-)$ when true. The satisfaction of a proposition that holds at time $i$ is immediately witnessed and 
it cannot be violated by any suffix. Similarly, a proposition evaluates to $(-,0)$ when false. The valuations of $\eventually g$ 
count the number of steps to positions in which $g$ holds. For instance, the first time at which $g$ holds is $i=3$, hence 
$\eventually g$  evaluates to $(2,-)$ at time $1$, $(1,-)$ at time $2$ and $(0,-)$ at time $3$. We also note that 
$\eventually g$ evaluates to $(0,\infty)$ at the end of the trace -- it could be immediately satisfied with the continuation of the 
trace with $g$ that holds, but could be violated only by an infinite suffix in which $g$ never holds. We finally observe that 
$\always(r \rightarrow \eventually g)$ evaluates to $(\infty,\infty)$ at all positions -- the property can be both satisfied and violated 
only with infinite suffixes.

\begin{table}[tb]
\centering
\caption{Unbounded response property example: $\countsem(\phi, i)$ with the trace $\pi = \tau_{2}$.}
   {\scriptsize
\begin{tabular}{r|cccccccc}
\toprule
	 &  1	& 2	& 3	& 4	& 5	& 6	& 7	  & EOT \\
\midrule

$r$    & $\top$ & $-$ & $-$ & $\top$ & $-$ & $-$ & $\top$ & \\
$g$     &  $-$ & $-$ & $\top$ & $-$ & $-$ & $\top$ & $-$ &  \\
\bottomrule
$\countsem(r,i)$ & $(0,-)$ & $(-,0)$ & $(-,0)$ & $(0,-)$ & $(-,0)$ & $(-,0)$ & $(0,-)$ &  (0,0) \\
$\countsem(g,i)$ & $(-,0)$ & $(-,0)$ & $(0,-)$ & $(-,0)$ & $(-,0)$ & $(0,-)$ & $(-,0)$ &  (0,0) \\
$\countsem(\neg r,i)$ & $(-,0)$ & $(0,-)$ & $(0,-)$ & $(-,0)$ & $(0,-)$ & $(0,-)$ & $(-,0)$ & (0,0) \\
$\countsem(\eventually g,i)$ & $(2,-)$ & $(1,-)$ & $(0,-)$ & $(2,-)$ & $(1,-)$ & $(0,-)$ & $(1,\infty)$ &  $(0,\infty)$ \\
$\countsem(r \rightarrow \eventually g,i)$ & $(2,-)$ & $(0,-)$ & $(0,-)$ & $(2,-)$ & $(0,-)$ & $(0,-)$ & $(1,\infty)$ & $(0,\infty)$ \\
$\countsem(\always(r \rightarrow \eventually g),i)$ & $(\infty,\infty)$ & $(\infty,\infty)$ & $(\infty,\infty)$  & $(\infty,\infty)$ & $(\infty,\infty)$ & $(\infty,\infty)$ & $(\infty,\infty)$ & $(\infty,\infty)$ \\
\bottomrule
\end{tabular}
}
\label{tab:response-sem}
\end{table}
\end{example}

Not all pairs $(s,f) \in \mathbb{N}_{+} \times \mathbb{N}_{+}$ are possible according to the counting semantics. We 
present the possible pairs in Lemma~\ref{lemma:admissible}.

\begin{lemma}
\label{lemma:admissible}
Let $\pi \in \Pi^{*}$ be a finite trace, $\phi$ an LTL formula and $i\in \mathbb{N}_0$ an index. We have that  
$\countsem(\phi,i)$ is of the form $(a,-)$, $(-,a)$, $(b_{1}, b_{2})$, $(b_{1},\infty)$, $(\infty, b_{2})$ or 
$(\infty, \infty)$, where $a \leq |\pi| - i$ and $b_{j} > |\pi| - i$ for $j \in \{1,2\}$.





\end{lemma}
\begin{proof} The proof can be obtained using structural induction on the LTL formula (see Appendix~\ref{proof:lemma:admissible}).
\end{proof}

Finally, we relate our counting semantics to the three valued semantics in Lemma~\ref{prop:comp}.

\begin{lemma}
\label{prop:comp}
Given an LTL formula and a trace $\pi \in \Pi^{*}$ where $i \in \setNpos$ is an index and $\phi$ is an LTL formula, we have that
$$
\begin{array}{lll}
\threeval(\phi, i) = \top & \leftrightarrow & \countsem(\phi,i) = (a,\impossible)\text{,} \\
&& \text{and }  \not \exists x < a \scope \pi' = \pi_i \cdot \pi_{i+1} \cdot \dots \pi_{i+x}, \mu_{\pi'}( \phi,1) = \top \\
\threeval(\phi,i) = \bot & \leftrightarrow & \countsem(\phi,i) = (\impossible,a)\text{,} \\
&& \text{and }  \not \exists x < a \scope \pi' = \pi_i \cdot \pi_{i+1} \cdot \dots \pi_{i+x}, \mu_{\pi'}( \phi,1) = \bot \\
\threeval(\phi,i) = \inconclusive & \leftrightarrow & \countsem(\phi,i)  = (b_{1}, b_{2})\text{,} \\
\end{array} 
$$
\noindent where $a \leq |\pi| - i$ and $b_{j}$ is either $\infty$ or $b_{j} > |\pi| - i$ for $j \in \{1,2\}$.
\end{lemma}

Intuitively, Lemma~\ref{prop:comp} holds because we only introduce the symbol ``$\impossible$'' within the trace when a satisfaction (violation) is observed. And the values of a pair only propagate into the past (and never into the future).

\subsection{Evaluation}
\label{sec:evaluation}

We now propose a mapping that predicts a qualitative verdict from our counting semantics. We adopt a 
$5$-valued set consisting of $\textsf{true}$ ($\top$), $\textsf{presumably true}$ ($\ptrue$),  $\textsf{inconclusive}$ ($\inconclusive$),
$\textsf{presumably false}$ ($\pfalse$) and $\textsf{false}$ ($\bot$) verdicts. We define the following 
order over these five values: $\bot < \pfalse < \inconclusive < \ptrue < \top$. We equip this $5$-valued domain with 
the negation ($\neg$) and disjunction ($\vee$) operations, letting $\neg \top = \bot$, $\neg \ptrue = \pfalse$, 
$\neg \inconclusive = \inconclusive$, $\neg \pfalse = \ptrue$, $\neg \bot = \top$ and 
$\phi_{1} \vee \phi_{2} = \max \{ \phi_{1}, \phi_{2} \}$. We define other Boolean operators such as conjunction 
 by the usual logical equivalences ($\phi_{1} \wedge \phi_{2} = \neg(\neg \phi_{1} \vee \neg \phi_{2})$, etc.).

%
%

We evaluate a property on a trace to $\top$ ($\bot$) when the satisfaction (violation) can be fully determined from the trace, following the definition of the three-valued semantics
$\mu$. Intuitively, this takes care of the case in which the safety
(co-safety) part of a formula has been violated (satisfied), at least
for properties that are intentionally safe (intentionally co-safe,
resp.)  \cite{DBLP:journals/fmsd/KupfermanV01}.

Whenever the truth value is not determined, we distinguish whether $d_\pi(\phi,i)$ indicates the possibility for a  satisfaction, respective violation, in finite time or not.
For possible satisfactions, respective violations, in finite time we make a prediction on whether past observations support the believe that the trace is going to satisfy or violate the property.
If the predictions are not inconclusive and not contradicting, then we evaluate the trace to the (presumable) truth value $\ptrue$ or$\pfalse$.
If we cannot make a prediction to a truth value, we compute the truth value recursively based on the operator in the formula and the truth values of the subformulas (with temporal operators unrolled). 
 
We use the predicate $\prediction_\pi$  to give the prediction based on the observed witnesses for satisfaction. 
The predicate $\prediction_\pi(\phi,i)$ becomes $\inconclusive$ when no witness for satisfaction exists in the past.
When there exists a witness that requires at least the same amount of additional steps as the trace under evaluation then the predicate evaluates to $\top$. If all the existing witnesses (and at least one exists)  are shorter than the current trace, then the predicate evaluates to $\bot$.
For a prediction on the violation we make a prediction on the satisfaction of $d_\pi(\neg \phi,i)$, i.e., we compute $\prediction_\pi(\neg \phi,i)$.

\begin{definition}[Prediction predicate]
\label{sec:finiteLTL:sec:countSem:def:predPredicate}

Let $s,f$ denote natural numbers and let $s_\pi(\phi,i), f_\pi(\phi,i) \in \mathbb{N}_+$ such that $\countsem(\phi,i)=\big(s_\pi(\phi,i),f_\pi(\phi,i) \big)$.
We define the $3$-valued predicate $\prediction_\pi$  as

\begin{align*}
\prediction_\pi(\phi,i) &= \begin{cases} 
    \top & \text{if } \exists j < i \scope \countsem(\phi,j) = (s',-) \text { and } s_\pi(\phi,i)  \leq s' \text{,}\\
    ? & \text{if } \not \exists j < i \scope \countsem(\phi,j) =
    (s',-) \text{,}\\
    \bot & \text{if } \exists j < i \scope \countsem(\phi,j) = (s',-) \text { and } \text{,}\\
     &  \mbox{ }  s_\pi(\phi,i) > \max_{0 \leq j < i} \{ s'~|~\countsem(\phi, j) = (s',-) \}  \text{,}\\
    \end{cases}  \\    
\end{align*}
\end{definition}

For the evaluation we consider a case split among the possible combinations of values in the pairs.

\begin{definition}[Predictive evaluation]
\label{def:predictive}
We define the {\em predictive evaluation} function $\predictive(\phi,i)$, with $a \leq |\pi| - i$ and $b_{j} > |\pi| - i$ for $j \in \{1,2\}$ and $a,b_j \in \mathbb{N}_0$, for the different cases of $ d_\pi(\phi,i)$:
\begin{center}
\begin{tabular}{llc}
$d_\pi(\phi,i) $ &  & $\predictive(\phi,i)$\\
\toprule
$(a,-)$&   &$\top $\\
\midrule   
            &$\text{if } \prediction_\pi(\phi,i) >  \prediction_\pi(\neg \phi,i) $ &   $ \ptrue$   \\
$  (b_1,b_2) $& $ \text{if }   \prediction_\pi(\phi,i) = \prediction_\pi(\neg \phi,i) $  & $\fsem(\phi,i)$ \\
            & $ \text{if }  \prediction_\pi(\phi,i) < \prediction_\pi(\neg \phi,i) $  &$ \pfalse$  \\ 
\midrule   
            &  $\text{if } \prediction_\pi(\phi,i) = \top $  & $ \ptrue$   \\
$  (b_1,\infty) $& $\text{if }    \prediction_\pi(\phi,i) = \inconclusive$  & $\fsem(\phi,i)$  \\
            & $ \text{if }   \prediction_\pi(\phi,i) = \bot$  &$ \pfalse$ \\ 
\midrule
$ (\infty,b_1) $&   & $\predictive(\neg \phi,i) $\\     
\midrule
$  (\infty,\infty)$ &   & $\fsem(\phi,i)  $\\
\midrule
$ (-,a)$&   &$ \bot$ \\
\bottomrule
\end{tabular}
\end{center}
\noindent where $\fsem(\phi,i)$ is an auxiliary function defined inductively as follows:

\begin{align*}
\fsem(p,  i)				& =	 \inconclusive \\
\fsem(\neg \phi,  i)			& =	 \neg \predictive(\phi, i) \\
\fsem(\phi_{1} \vee \phi_{2},  i)	& =	 \predictive(\phi_{1}, i) \vee \predictive(\phi_{2},  i) \\
\fsem(\nextt^n \phi, i)			& =	 \predictive(\phi,i+n) \\
\fsem(\eventually \phi, i)	& =
  \begin{cases}
   \predictive(\phi,  i) \vee\fsem(\nextt   \eventually \phi, i)  & \text{if } i \leq |\pi| \\
    \predictive(\phi,  i)  & \text{if } i > |\pi| \\
    \end{cases}  \\
\fsem(\phi_{1} \until \phi_{2}, i)	& =	
  \begin{cases}
   \predictive(\phi_{2},  i) \vee (\predictive(\phi_{2}, i) \wedge \predictive(\nextt (\phi_{1} \until \phi_{2}),  i)  & \text{if } i \leq |\pi| \\
   \predictive(\phi_{2},  i)  & \text{if } i > |\pi| \\
    \end{cases}  
\end{align*}
\end{definition}

The predictive evaluation function is symmetric. Hence, $\predictive(\phi,i) = \neg \predictive(\neg \phi,i)$ holds.

\begin{example}
The outcome of evaluating $\tau_2$ from Table~\ref{tab:response} is
shown in Table~\ref{tab:response-eval}. Subformula $r \rightarrow
\eventually g$ is predicted to be $\ptrue$ at $i=7$ because there
exists a longer witness for satisfaction in the past (e.g., at $i=1$).
Thus, the trace evaluates to $\ptrue$, as expected.
\begin{table}[tb]
\centering
\caption{Unbounded response property example with $\pi = \tau_{2}$.}
   {\scriptsize
\begin{tabular}{r|cccccccc}
\toprule
	 &  1	& 2	& 3	& 4	& 5	& 6	& 7	  & EOT \\
\midrule
$r$    & $\top$ & $-$ & $-$ & $\top$ & $-$ & $-$ & $\top$ & \\
$g$     &  $-$ & $-$ & $\top$ & $-$ & $-$ & $\top$ & $-$ &  \\
\bottomrule
$\countsem(r,i)$ & $(0,-)$ & $(-,0)$ & $(-,0)$ & $(0,-)$ & $(-,0)$ & $(-,0)$ & $(0,-)$ &  (0,0) \\
$\predictive(r,i) $ & $\top$  & $\bot$  & $\bot$  & $\top$  & $\bot$  & $\bot$  & $\top$ & ?\\
\midrule
$\countsem(g,i)$ & $(-,0)$ & $(-,0)$ & $(0,-)$ & $(-,0)$ & $(-,0)$ & $(0,-)$ & $(-,0)$ &  (0,0) \\
$\predictive(g,i) $ & $\bot$  & $\bot$  & $\top$  & $\bot$  & $\bot$  & $\top$  & $\bot$ & ?\\
\midrule
$\countsem(\eventually g,i)$ & $(2,-)$ & $(1,-)$ & $(0,-)$ & $(2,-)$ & $(1,-)$ & $(0,-)$ & $(1,\infty)$ &  $(0,\infty)$ \\
$\predictive(\eventually g,i) $ & $\top$ & $\top$ & $\top$ & $\top$ & $\top$ & $\top$ & $\ptrue$ &   $\ptrue$\\
\midrule
$\countsem(r \rightarrow \eventually g,i)$ & $(2,-)$ & $(0,-)$ & $(0,-)$ & $(2,-)$ & $(0,-)$ & $(0,-)$ & $(1,\infty)$ & $(0,\infty)$ \\
$\predictive(r \rightarrow \eventually g,i) $  & $\top$  & $\top$  & $\top$  & $\top$  & $\top$  & $\top$  & $\ptrue$ &   $\ptrue$\\
\midrule
$\countsem(\always(r \rightarrow \eventually g),i)$ & $(\infty,\infty)$ & $(\infty,\infty)$ & $(\infty,\infty)$  & $(\infty,\infty)$ & $(\infty,\infty)$ & $(\infty,\infty)$ & $(\infty,\infty)$ & $(\infty,\infty)$ \\
$\predictive(\always(r \rightarrow \eventually g),i) $ &   $\ptrue$ &   $\ptrue$ &   $\ptrue$ &   $\ptrue$ &   $\ptrue$ &   $\ptrue$ &   $\ptrue$ &   $\ptrue$\\
\bottomrule
\end{tabular}
}
\label{tab:response-eval}
\end{table}
\end{example}



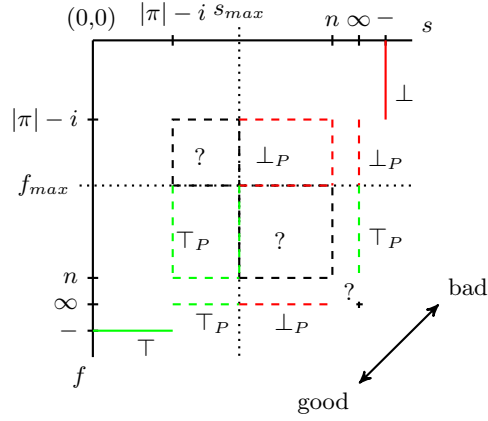
\begin{figure}
  \centering
  \begin{tikzpicture}[thick, scale=0.7]

  \draw [-] (0,0) -- (0,-6);
  \node [below] at (-0.25,-6) {$f$};
  \draw [-] (-0.1,-5.5) -- (0.1,-5.5);
  \node [left] at (-0.1,-5.5) {$-$};
  \draw [-] (-0.1,-5) -- (0.1,-5);
  \node [left] at (-0.1,-5) {$\infty$};
  \draw [-] (-0.1,-4.5) -- (0.1,-4.5);
  \node [left] at (-0.1,-4.5) {$n$};
  \draw [-] (-0.1,-1.5) -- (0.1,-1.5);
  \node [left] at (-0.1,-1.5) {$|\pi| - i$};
  
  \draw [-] (0,0) -- (6,0);
  \node [right] at (6,0.25) {$s$};
  \draw [-] (5.5,-0.1) -- (5.5,0.1);
  \node [above] at (5.5,0.1) {$-$};
  \draw [-] (5,-0.1) -- (5,0.1);
  \node [above] at (5,0.1) {$\infty$};
  \draw [-] (4.5,-0.1) -- (4.5,0.1);
  \node [above] at (4.5,0.1) {$n$};
  \draw [-] (1.5,-0.1) -- (1.5,0.1);
  \node [above] at (1.5,0.1) {$|\pi| - i$};  
  
  \node [above] at (-.025,0) {(0,0)};

  \node [above] at (2.75,0.25) {$s_{max}$};
  \draw [dotted,-] (2.75,0.25) -- (2.75,-6);
  \node [left] at (-0.25,-2.75) {$f_{max}$};
  \draw [dotted,-] (-0.25,-2.75) -- (6,-2.75);

  
  \node [left] at (2.3,-2.2) {$?$};
  \draw[dashed] (1.5,-1.5) rectangle (2.75,-2.75);
  \node [left] at (3.8,-3.8) {$?$}; 
  \draw[dashed] (2.75,-2.75) rectangle (4.5,-4.5); 

  \node [left] at (3.9,-2.2) {$\bot_P$};
  \draw[dashed, color=red,-] (2.75,-1.5) -- (4.5,-1.5);
  \draw[dashed, color=red,-] (4.5,-1.5) -- (4.5,-2.75);
  \draw[dashed, color=red,-] (2.75,-2.75) -- (4.5,-2.75);
  \node [left] at (2.4,-3.8) {$\top_P$};
  \draw[dashed, color=green,-] (1.5,-2.75) -- (1.5,-4.5);
  \draw[dashed, color=green,-] (2.75,-2.75) -- (2.75,-4.5);
  \draw[dashed, color=green,-] (1.5,-4.5) -- (2.75,-4.5);
  
  \draw [thick,-,color=green] (0,-5.5) -- (1.5,-5.5);
  \node [below] at (1,-5.5) {$\top$};
  \draw [dashed,-,color=green] (1.5,-5) -- (2.75,-5);
  \node [below] at (2.25,-5) {$\top_P$};
  \draw [dashed,-,color=red] (2.75,-5) -- (4.5,-5);
  \node [below] at (3.75,-5) {$\bot_P$};
  \draw [thick,-,color=red] (5.5,0) -- (5.5,-1.5);
  \node [right] at (5.5,-1) {$\bot$};
  \draw [dashed,-,color=red] (5,-1.5) -- (5,-2.75);
  \node [right] at (5,-2.25) {$\bot_P$};
  \draw [dashed,-,color=green] (5,-2.75) -- (5,-4.5);
  \node [right] at (5,-3.75) {$\top_P$};
  
  \draw [-] (5,-4.95) -- (5,-5.07);
  \draw [-] (4.95,-5) -- (5.07,-5);
  \node [above left] at (5.1,-5.1) {$?$};
  
  \draw [<->] (5,-6.5) -- (6.5,-5);  
  \node [below left] at (5,-6.5) {\text{good}};
  \node [above right] at (6.5,-5) {\text{bad}};
\end{tikzpicture}
  \caption{Lattice for $(s,f)$ with $\phi$ and $i < |\pi|$ fixed.}
  \label{sec:finiteLTL:sec:lattice}
\end{figure}


In Figure~\ref{sec:finiteLTL:sec:lattice} we visualize the evaluation of a pair $\countsem(\phi,i) = (s,f)$ for a fixed $\phi$ and a fixed position $i$. 
On the x-axis is the witness count $s$ for a satisfaction and on the y-axis is the witness count $f$ for a violation.
For a value $s$, respectively $f$, that is smaller than the length of the suffix starting at position $i$ (with the other value of the pair always being $\impossible$), the evaluation is either $\top$ or $\bot$.
Otherwise the evaluation depends on the values $s_{max}$ and $f_{max}$. These two values represent the largest witness counts for a satisfaction and a violation in the past, i.e., for positions smaller than $i$  in the trace.
Based on the prediction function $\prediction_\pi(\phi,i)$ the evaluation becomes  $\top_P$, $\inconclusive$ or $\bot_P$, where $\inconclusive$ indicates that the auxiliary function $\fsem(\phi,i)$ has to be applied. 
Starting at an arbitrary point in the diagram and moving to the right increases the witness count for a satisfaction while the witness count for a violation remains constant. Thus, moving to the right makes the pair ``more false''. The same holds when keeping the witness count for a satisfaction constant and moving up in the diagram as this decrease the witness count for a violation. Analogously, moving down and/or left makes the pair ``more true'' as the witness count for a violation gets larger and/or the witness count for a satisfaction gets smaller.

Our $5$-valued predictive evaluation refines the $3$-valued LTL semantics.

\begin{theorem}
\label{theorem:refinement}
Let $\phi$ be an LTL formula, $\pi \in \Pi^{*}$ and $i \in \setNpos$. We have 
$$
\begin{array}{lllll}
\threeval(\phi,i) = & \top & \leftrightarrow & \predictive(\phi, i) = & \top \text{,}\\ 
\threeval(\phi,i) = & \bot & \leftrightarrow & \predictive(\phi, i) = & \bot \text{,}\\
\threeval(\phi,i) = & \inconclusive & \leftrightarrow & \predictive(\phi, i) \in & \{ \ptrue, \pfalse, \inconclusive \} \text{.}\\
\end{array}
$$
\end{theorem}

Theorem~\ref{theorem:refinement} holds, because the evaluation to $\top$ and $\bot$ is simply the mapping of a pair that contains the symbol ``$\impossible$'', which we have shown in Lemma~\ref{prop:comp}. 

Remember that $\mathbb{N}_+ \times \mathbb{N}_+$ is partially ordered
by $\unlhd$. We now show that having a trace that is ``more true'' than another is correctly reflected in our finitary semantics.
To define ``more true'', we first need the polarity of a proposition in an LTL formula.

\begin{example}
Note that $g$ has positive polarity in $\phi = \always(r \rightarrow
\eventually g)$. If we define $\tau'_2$ to be as $\tau_2$, except that
$g \in \tau'_2(i)$ for $i \in \{1,\dots,6\}$, we have
$e_{\tau'_2}(\phi,i) = \pfalse$, whereas $e_{\tau_2}(\phi,i) = \ptrue$.
\end{example}

\begin{definition}[Polarity]
Let $\# \neg$ be the number of negation operators on a specific path in the parse tree of $\phi$ starting at the root.
We define the polarity as the function $\text{pol}(p)$ with proposition $p$ in an LTL formula $\phi$ as follows:
\begin{align*}
\text{pol}(p) = \begin{cases}
    \text{pos,} & \text{if $\# \neg$ on all paths to a leaf with proposition $p$ is even,}\\
    \text{neg,} & \text{if $\# \neg$ on all paths to a leaf with proposition $p$ is odd,}\\
    \text{mixed, }& \text{otherwise.}\\
  \end{cases}
\end{align*}
\end{definition}

With the polarity defined, we now define the constraints for a trace to be ``more true'' with respect to an LTL formula $\phi$.

\begin{definition}[$\pi \sqsubseteq_\phi \pi'$]
Given two traces $\pi$ and $\pi'$ of equal length and an LTL formula $\phi$ over proposition $p$, we define that $\pi \sqsubseteq_\phi \pi'$ iff
$$
\begin{array}{lllll}
\forall i \forall p \scope &  pol(p) = \text{mixed} \Rightarrow p \in \pi_i \leftrightarrow p \in \pi'_i \text{ and}\\
   &pol(p) = \text{pos} \Rightarrow p \in \pi_i \rightarrow p \in \pi'_i \text{ and}\\
   &pol(p) = \text{neg} \Rightarrow p \in \pi_i \leftarrow p \in \pi'_i. \\
\end{array}
$$
\end{definition}

Whenever one trace is ``more true'' than another, this is correctly reflected in our finitary semantics.

\begin{theorem}
\label{theorem:moreTrue}
For two traces $\pi$ and $\pi'$ of equal length and an LTL formula $\phi$ over proposition $p$, we have that
$$
\pi \sqsubseteq_\phi \pi' \Rightarrow d_{\pi'}(\phi,1) \unlhd d_{\pi}(\phi,1).
$$
Therefore, we have for $\pi \sqsubseteq_\phi \pi'$ that
\begin{align*}
e_{\pi}(\phi, 1) = \top  \Rightarrow e_{\pi'}(\phi, 1) = \top& \text{, and}\\
e_{\pi}(\phi, 1) = \bot  \Leftarrow e_{\pi'}(\phi, 1) = \bot&.
\end{align*}
\end{theorem}

Theorem~\ref{theorem:moreTrue} holds,
because we have that replacing an arbitrary observed value in $\pi$ by one with positive polarity in $\pi'$ always results  with  $d_{\pi}(\phi,1)=(s,f)$ and $d_{\pi'}(\phi,1)=(s',f')$ in $s' \leq s$ and $f' \geq f$, as with $\pi \sqsubseteq_\phi \pi'$ we have that $\pi'$ witnesses a satisfaction of $\phi$ not later than $\pi$ and $\pi'$ also witness a violation of $\phi$ not earlier than $\pi$.

\begin{table}[tb]
\center
\caption{Making a system ``more true''.}
  \begin{tabular}{c | c |  c | c}
      $\phi$ & $\pi$ & $\countsem(\phi,1)$ & $\predictive(\phi,1)$ \\
    \toprule
      \multirow{2}{*}{$p$} & $-$ &$ (-,0)$ & $\bot$ \\
      & $\top$ &$ (0,-) $& $\top$ \\
    \midrule
      \multirow{2}{*}{$ p \land \nextt \eventually p$} & $- - -$ &$ (-,0) $& $\bot$ \\
      & $\top - - $ &$(3,\infty)$ & $\pfalse$ \\
    \midrule
      \multirow{2}{*}{$\always p$} &  $- \top \top$ &$ (-,0) $& $\bot$ \\
      & $\top \top \top$ & $(\infty,3)$ & $\ptrue$ \\
    \midrule
      \multirow{2}{*}{$\eventually p$} &  $- - -$ &$ (3,\infty) $& $\pfalse$ \\
      & $\top - -$ & $(0,-)$ & $\top$ \\
    \midrule
      \multirow{2}{*}{$\eventually \always p$} &  $\top - \top - \top$ &$ (\infty,\infty) $& $\pfalse$ \\
      &$\top - \top \top \top$  & $ (\infty,\infty)$ & $\ptrue$ \\
    \midrule
      \multirow{2}{*}{$\always \eventually p$} &  $- - \top - -$ &$ (\infty,\infty) $& $\ptrue$ \\
      & $\top - \top - -$ & $ (\infty,\infty)$ & $\pfalse$ \\
    \midrule
      \multirow{2}{*}{$p \lor \nextt \always p$} &  $- \top \top$ &$ (\infty,3) $& $\ptrue$ \\
      & $\top \top \top$ & $(0,-)$ & $\top$ \\
    \bottomrule
  \end{tabular}
\label{tab:examplesMoreTrue}
\end{table}

In Table~\ref{tab:examplesMoreTrue} we give examples to illustrate the transition of one evaluation to another one. Note that it is possible to change from $\ptrue$ to $\pfalse$. However, this is only the predicated truth value that becomes ``worse'', because we have strengthened the prefix on which the prediction is based on, the values of $\countsem(\phi,i)$  don't change and remain the same is such a case.  

\section{Examples}


We demonstrate the strengths and weaknesses of our approach on the examples of LTL specifications and traces shown in Table~\ref{tab:ltl-example}. 
We fully develop these examples in Appendix~\ref{sec:examples}.  

\begin{table}
\centering

\begin{tabular}{| l c l | l l l |  l l l | }

\midrule

\multicolumn{3}{|c|}{Specifications} & \multicolumn{6}{|c|}{Traces} \\\midrule

$\psi_{1}$	& $\equiv$	& $\eventually \nextt g$  									& $\pi_{1}:$	& $g:$ 	& $\bot \bot \bot \bot$			&		
$\pi_{5}:$	& $r:$		& $\bot \top \top \top \top \bot \top \top$ 						\\

$\psi_{2}$	& $\equiv$	& $\always \nextt g$ 									& $\pi_{2}$	& $g:$	& $\top \top \top \top$			&
		& $g:$	& $\bot \top \bot \bot \bot \bot \top \bot$ 						\\

$\psi_{3}$	& $\equiv$	& $\always (r \rightarrow \eventually g)$ 						& $\pi_{3}$	& $r:$		& $\bot \top \bot \bot \top \bot$	&
$\pi_{6}:$	& $g:$	& $\top \top \bot \bot \top \top \bot \bot \top \top \bot \bot \top$			\\

$\psi_{4}$	& $\equiv$	& $\bigwedge_{i \in \{ 1,2\}} \always (r_{i} \rightarrow \eventually g_{i})$	&		& $g:$	& $\bot \bot \top \bot \bot \bot$	&
$\pi_{7}:$	& $g:$	& $\top \top \bot \bot \top \top \bot \bot \top \top \top \top \top$			\\

$\psi_{5}$	& $\equiv$	& $\always ((\nextt r) \until (\nextt \nextt g))$ 					& $\pi_{4}$	&$ r_{1}:$	& $\top \bot \top \bot \top \bot \top$		&	
$\pi_{8}$	& $r:$		& $\top \top \top \top \bot \bot$								\\

$\psi_{6}$	& $\equiv$	& $\eventually \always g \vee \eventually \always \neg g$ 				&		& $g_{1}:$	& $\bot \top \bot \top \bot \top \bot$  		&
		& $g:$	& $\top \bot \top \bot \top \bot$								\\

$\psi_{7}$	& $\equiv$	& $\always(\eventually r \vee \eventually g)$ 						& 		& $r_{2}:$	& $\bot \top \bot \top \bot \top \bot$	&
	& 		& 						\\

$\psi_{8}$	& $\equiv$	& $\always \eventually (r \vee g)$							&		& $g_{2}:$	& $\top \bot \top \bot \top \bot \top$ 	&
		                                                                                                       && 								\\	

$\psi_{9}$	& $\equiv$	& $\always \eventually r \vee \always \eventually g$ 					& 		&	& 	&
		&		&												\\

\midrule
\end{tabular}
\caption{Examples of LTL specifications and traces} 
\label{tab:ltl-example}
\end{table}

Table~\ref{tab:ltl-verdicts} summarizes the evaluation of our examples. The first and the second column denote the evaluated specification and trace. 
We use these examples to compare LTL with counting semantics (\textsf{c-LTL}) presented in this paper, to the other two popular finitary LTL interpretations, 
the $3$-valued LTL semantics~\cite{DBLP:conf/fsttcs/BauerLS06} (\textsf{3-LTL}) and 
LTL on trucated paths~\cite{DBLP:conf/cav/EisnerFHLMC03} (\textsf{t-LTL}). We recall that in \textsf{t-LTL} there is a distinction between a weak and a strong next 
operator. We denote by \textsf{t-LTL-s} (\textsf{t-LTL-w}) the specifications from our examples in which $\nextt$ is interpreted as the strong (weak) next operator and 
assume that we always give a strong interpretation to $\until$ and $\eventually$ and a weak interpretation to $\always$. 

\begin{table}
\centering
\begin{tabular}{| c |  c |  c | c | c | c |}
\midrule
Spec. & Trace & \textsf{c-LTL} & \textsf{3-LTL} & \textsf{t-LTL-s} & \textsf{t-LTL-w} \\
\midrule      
$\psi_{1}$	 &  $\pi_{1}$ & $\bot_P$ & \inconclusive & $\bot$ & $\top$ \\
\midrule
$\psi_{2}$	 &  $\pi_{2}$ & $\top_P$ & \inconclusive & $\bot$ & $\top$ \\
\midrule
$\psi_{3}$	 &  $\pi_{3}$ & $\bot_P$ & \inconclusive & $\bot$ & $\bot$\\
\midrule
$\psi_{4}$	 &  $\pi_{4}$ & $\top_P$ & \inconclusive & $\bot$ & $\bot$\\
\midrule
$\psi_{5}$	 &  $\pi_{5}$ & $\top_P$ & \inconclusive & $\bot$ &  $\top$ \\
\midrule
$\psi_{6}$	 &  $\pi_{6}$ & $\bot_P$ & \inconclusive & $\top$ & $\top$\\
\midrule
$\psi_{6}$	 &  $\pi_{7}$ & $\top_P$ & \inconclusive & $\top$ & $\top$ \\
\midrule
$\psi_{7}$	 &  $\pi_{8}$ & $\bot_P$ & \inconclusive & $\bot$ & $\bot$ \\
\midrule
$\psi_{8}$	 &  $\pi_{8}$ & $\bot_P$ & \inconclusive & $\bot$ & $\bot$\\
\midrule
$\psi_{9}$	 &  $\pi_{8}$ & $\top_P$ & \inconclusive & $\bot$ & $\bot$\\
\midrule

\end{tabular}
\caption{Comparison of different verdicts with different semantics}
\label{tab:ltl-verdicts}
\end{table}

There are two immediate observations that we can make regarding the results presented in Table~\ref{tab:ltl-verdicts}. First, the $3$-valued LTL gives for 
all the examples an \emph{inconclusive} verdict, a feedback that after all has little value to a verification engineer. The second observation is that 
the verdicts from \textsf{c-LTL} and \textsf{t-LTL} can differ quite a lot, which is not very surprising given the different strategies to interpret the unseen 
future. We now further comment on these examples, explaining in more details the results and highlighting the intuitive outcomes  of \textsf{c-LTL} for a 
large class of interesting LTL specifications.  

\paragraph{Effect of Nested Next} We evaluate with $\psi_{1}$ and $\psi_{2}$ the effect of nesting $\nextt$ in an $\eventually$ and an $\always$ 
formula, respectively. We make a prediction on $\nextt g$ at the end of the trace before evaluating $\eventually$ and $\always$. As a consequence, 
we find that $(\psi_{1}, \pi_{1})$ evaluates to \textsf{presumably false}, while $(\psi_{2}, \pi_{2})$ evaluates to \textsf{presumably true}.
In \textsf{t-LTL}, this class of specification is very sensitive to the weak/strong interpretation 
of next, as we can see from the verdicts.

\paragraph{Request/Grants} We evaluate the request/grant property $\psi_{3}$ from the motivating example on the trace $\pi_{3}$. 
We observe that $r$ at cycle $2$ is followed by $g$ at cycle $3$, while $r$ at cycle $5$ is not followed by $g$ at cycle $6$. Hence, 
$(\psi_{3}, \pi_{3})$ evaluates to \textsf{presumably false}.

\paragraph{Concurrent Request/Grants} We evaluate the specification $\psi_{4}$ against the trace $\pi_{4}$.  In this example $r_{1}$ is triggered at 
even time stamps and $r_{2}$ is triggered at odd time stamps. Every request is granted in one cycle. It follows that regardless of the time when the trace 
ends, there is one request that is not granted yet. We note that $\psi_{4}$ is a conjunction of two basic request/grant properties and we make independent 
predictions for each conjunct. Every basic request/grant property is evaluated to \textsf{presumably true}, hence $(\psi_{4}, \pi_{4})$ evaluates to \textsf{presumably true}. 
At this point, we note that in \textsf{t-LTL}, every request that is not granted by the end of the trace results in the property violation, regardless of the 
past observations.

\paragraph{Until} We use the specification $\psi_{5}$ and the trace $\pi_{5}$ to evaluate the effect of $\until$ on the predictions. The specification 
requires that $\nextt r$ continuously holds until $\nextt \nextt g$ becomes true. We can see that in $\pi_{5}$ $\nextt r$ is witnessed at cycles 
$1-4$, while $\nextt \nextt g$ is witnessed at cycle $5$. We can also see that $\nextt r$ is again witnessed from cycle $6$ until the end of the trace at cycle $8$. 
As a consequence, $(\psi_{5}, \pi_{5})$ is evaluated to \textsf{presumably true}.

\paragraph{Stabilization} The specification $\psi_{6}$ says that the value of $g$ has to eventually stabilize to either true or false. We evaluate the formula on 
two traces $\pi_{6}$ and $\pi_{7}$. In the trace $\pi_{6}$, $g$ alternates between true and false every two cycles and becomes true in the last cycle. Hence, there is no 
sufficiently long witness of trace stabilization $(\psi_{6}, \pi_{6})$ evaluates to \textsf{presumably false}. In the trace $\pi_{7}$, $g$ also alternates between true and false 
every two cycles, but in the last four cycles $g$ remains continuously true. As a consequence, $(\psi_{6}, \pi_{7})$ evaluates to \textsf{presumably true}. This example also 
illustrates the importance of when the trace truncation occurs. If both $\pi_{6}$ and $\pi_{7}$ were truncated at cycle $5$, both $(\psi_{6}, \pi_{6})$ and 
$(\psi_{6}, \pi_{7})$ would evaluate to \textsf{presumably false}. We note that $\psi_{6}$ is satisfied by all traces in \textsf{t-LTL}.

\paragraph{Sub-formula Domination} The specification $\psi_{7}$ exposes a weakness of our approach. It requires that in every cycle, either $r$ or $g$ is witnessed in 
some unbounded future.  With our approach, $(\psi_{7}, \pi_{8})$ evaluates to \textsf{presumably false}. This is against our intuition because we have observed that $g$ 
becomes regularly true very second time step. However, in this example our prediction for $\eventually r$ dominates over the prediction for $\eventually g$, leading to 
the unexpected \textsf{presumably false} verdict. On the other hand, \textsf{t-LTL} interpretation of the same specification is dependent only on the last value of 
$r$ and $g$.

\paragraph{Semantically Equivalent Formulas} We now demonstrate that our approach may give different answers for semantically equivalent formulas. For instance, 
both $\psi_{8}$ and $\psi_{9}$ are semantically equivalent to $\psi_{7}$. We have that $(\psi_{8}, \pi_{8})$ evaluates to \textsf{presumably false}, while 
$(\psi_{9}, \pi_{8})$ evaluates to \textsf{presumably true}. We note that  \textsf{t-LTL} verdicts are stable for semantically different formulas. 

\section{Conclusion}

We have presented a novel finitary semantics for LTL
that uses the history of satisfaction and violation in a finite trace
to predict whether the co-safety and safety aspects of a formula will
be satisfied in the extension of the trace to an infinite one. We
claim that the semantics closely follow human intuition when
predicting the truth value of a trace.  The presented examples 
(incl. non-monitorable LTL properties) illustrate our approach 
and support this claim.

Our definition of the semantics is trace-based, but it is easily
extended to take an entire database of traces into account, which may
make the approach more precise. Our approach uses a very simple form
of learning to predict the future. It would be interesting to consider
more elaborate learning methods to make better predictions.

\noindent{}{\bf Acknowledgments.}

\noindent This work was partially supported by the European Union (IMMORTAL
project, grant no. 644905),  the Austrian FWF (National Research 
Network RiSE/SHiNE S11405-N23 and S11406-N23), the SeCludE 
 project funded by UnivPM.
This work has been conducted within the ENABLE-S3 project that 
has received funding from the ECSEL Joint Undertaking under 
Grant Agreement no. 692455. This Joint Undertaking receives 
support from the European UnionÕs HORIZON  2020 research and 
innovation programme and Austria, Denmark, Germany, Finland, 
Czech Republic, Italy, Spain, Portugal, Poland, Ireland, Belgium, 
France, Netherlands, United Kingdom, Slovakia, Norway.

\bibliographystyle{plain}
\bibliography{bibliography}
  
\appendix

\section{Proofs}

\subsection{Proof for Lemma~\ref{lemma:soundness}}
\label{proof:lemma:soundness}

\begin{proof} 

Let be $i \in \mathbb{N}_{>0}$, $i \leq n = |\pi|$, $n > 0$ and $\pi_{i \cdots n}$ a suffix of $\pi$ starting at position $i$.
\begin{description}

{\small \item[Base case $p$] 
$$ \threeval(p,i) = \top \Rightarrow p \in \pi_i,  i \leq |\pi|  \Rightarrow  \forall \pi' \in \Pi^{\omega}, \pi_i \cdot \pi' \models p \Rightarrow [\pi_{i \cdots n}  \models_{3} p] = \top $$
$$ \threeval(p,i) = \bot \Rightarrow p \not \in \pi_i, i \leq |\pi| \Rightarrow  \forall \pi' \in \Pi^{\omega}, \pi_i \cdot \pi' \not \models p \Rightarrow [\pi_{i \cdots n} \models_{3} p] = \bot $$

\item[Induction step for $\neg \varphi$] :
$$\underbrace{(\threeval(\varphi,i) = \top  \Rightarrow  [\pi_{i \cdots n} \models_{3} \varphi] = \top)}_\text{Assumption step} \Rightarrow \underbrace{\underbrace{(\threeval(\neg \varphi,i) = \bot}_\text{$\iff \threeval(\varphi,i) = \top $}}_\text{True by assumption}   \Rightarrow  \underbrace{\underbrace{[\pi_{i \cdots n} \models_{3}  \neg \varphi] = \bot)}_\text{$\iff [\pi_{i \cdots n} \models_{3} \varphi] = \top$}}_\text{True by assumption} $$

$$\underbrace{(\threeval(\varphi,i) = \bot  \Rightarrow  [\pi_{i \cdots n} \models_{3} \varphi] = \bot)}_\text{Assumption step} \Rightarrow \underbrace{\underbrace{(\threeval(\neg \varphi,i) = \top}_\text{$\iff \threeval(\varphi,i) = \bot$}}_\text{True by assumption}    \Rightarrow  \underbrace{\underbrace{[\pi_{i \cdots n} \models_{3}  \neg \varphi] = \top)}_\text{$\iff [\pi_{i \cdots n} \models_{3} \varphi] = \top$}}_\text{True by assumption} $$

%

\item[Induction step for $\varphi_1 \vee \varphi_2$] 

\[ 
\left(
    \begin{array}{c}
    \threeval(\varphi_1,i) = \bot  \Rightarrow  [\pi_{i \cdots n} \models_{3} \varphi_1] = \bot  \\
    \threeval(\varphi_2,i) = \bot  \Rightarrow  [\pi_{i \cdots n} \models_{3} \varphi_2] = \bot 
    \end{array}
  \right) \Rightarrow   \threeval(\varphi_1 \vee \varphi_2,i) = \bot  \Rightarrow  [\pi_{i \cdots n} \models_{3} \varphi_1 \vee \varphi_2] = \bot 
\]

\[ 
\left(
    \begin{array}{c}
    \threeval(\varphi_1,i) = \top  \Rightarrow  [\pi_{i \cdots n} \models_{3} \varphi_1] = \top \\
    \threeval(\varphi_2,i) = \top  \Rightarrow  [\pi_{i \cdots n} \models_{3} \varphi_2] = \top 
    \end{array}
  \right) \Rightarrow   \threeval(\varphi_1 \vee \varphi_2,i) = \top  \Rightarrow  [\pi_{i \cdots n} \models_{3} \varphi_1 \vee \varphi_2] = \top 
\]

\[ 
\left(
    \begin{array}{c}
    \threeval(\varphi_1,i) = \top  \Rightarrow  [\pi_{i \cdots n} \models_{3} \varphi_1] = \top \\
    \threeval(\varphi_2,i) = \bot  \Rightarrow  [\pi_{i \cdots n} \models_{3} \varphi_2] = \bot 
    \end{array}
  \right) \Rightarrow   \threeval(\varphi_1 \vee \varphi_2,i) = \top  \Rightarrow  [\pi_{i \cdots n} \models_{3} \varphi_1 \vee \varphi_2] = \top 
\]

\[ 
\left(
    \begin{array}{c}
    \threeval(\varphi_1,i) = \bot  \Rightarrow  [\pi_{i \cdots n} \models_{3} \varphi_1] = \bot \\
    \threeval(\varphi_2,i) = \top  \Rightarrow  [\pi_{i \cdots n} \models_{3} \varphi_2] = \top 
    \end{array}
  \right) \Rightarrow   \threeval(\varphi_1 \vee \varphi_2,i) = \top  \Rightarrow  [\pi_{i \cdots n} \models_{3} \varphi_1 \vee \varphi_2] = \top 
\]

In the first case we have that: 
$$\threeval(\varphi_1 \vee \varphi_2,i) = \bot \iff (\threeval(\varphi_1,i) = \bot \wedge \threeval(\varphi_2,i) = \bot ) $$ 
$$\Rightarrow \underbrace{(  [\pi_{i \cdots n} \models_{3} \varphi_1] = \bot )\wedge ( [\pi_{i \cdots n} \models_{3} \varphi_2] = \bot))}_\text{True by assumption} \Rightarrow  [\pi_{i \cdots n} \models_{3} \varphi_1 \vee \varphi_2] = \bot  $$

For the other three cases we have: 

$$\threeval(\varphi_1 \vee \varphi_2,i) = \top \iff 
\left(
    \begin{array}{l}
    (\threeval(\varphi_1,i) = \top \wedge \threeval(\varphi_2,i) = \top ) \; \vee  \\
    (\threeval(\varphi_1,i) = \top \wedge \threeval(\varphi_2,i) = \bot ) \; \vee  \\
    (\threeval(\varphi_1,i) = \bot \wedge \threeval(\varphi_2,i) = \top ) 
    \end{array}
  \right)
$$ 
$$\Rightarrow \underbrace{
\left(
    \begin{array}{l}
    (  [\pi_{i \cdots n} \models_{3} \varphi_1] = \top )\wedge ( [\pi_{i \cdots n} \models_{3} \varphi_2] = \top)) \; \vee  \\
     (  [\pi_{i \cdots n} \models_{3} \varphi_1] = \top )\wedge ( [\pi_{i \cdots n} \models_{3} \varphi_2] = \bot)) \; \vee  \\
     (  [\pi_{i \cdots n} \models_{3} \varphi_1] = \bot )\wedge ( [\pi_{i \cdots n} \models_{3} \varphi_2] = \top)) 
    \end{array}
  \right)}_\text{True by assumption} \Rightarrow  [\pi_{i \cdots n} \models_{3} \varphi_1 \vee \varphi_2] = \top  $$

\item[Induction step for $\nextt \varphi$] We can prove that for $ i+1 \leq n$:
$$\underbrace{(\threeval(\varphi,i+1) = \top  \Rightarrow   [\pi_{i+1 \cdots n}  \models_{3} \varphi] = \top)}_\text{Assumption step} \Rightarrow \underbrace{\underbrace{(\threeval(\nextt \varphi,i) = \top}_\text{$\iff \threeval(\varphi,i+1) = \top $}}_\text{True by assumption}  \Rightarrow \underbrace{ \underbrace{[\pi_{i \cdots n}  \models_{3}  \nextt \varphi] = \top)}_\text{$\iff [\pi_{i+1 \cdots n}  \models_{3} \varphi] = \top)$}}_\text{True by assumption} $$

$$\underbrace{(\threeval(\varphi,i+1) = \bot  \Rightarrow  [\pi_{i+1 \cdots n}  \models_{3} \varphi] = \bot)}_\text{Assumption step}  \Rightarrow  \underbrace{ \underbrace{(\threeval(\nextt \varphi,i) = \bot}_\text{$\iff \threeval(\varphi,i+1) = \bot $}}_\text{True by assumption}    \Rightarrow 
 \underbrace{ \underbrace{  [\pi_{i \cdots n}  \models_{3}  \nextt \varphi] = \bot)}_\text{$[\pi_{i+1 \cdots n}  \models_{3} \varphi] = \bot$}}_\text{True by assumption} $$
%
%

\item[Induction step for $\eventually \varphi$] We  can prove that $\exists j, i \leq j \leq |\pi|$:
$$\underbrace{(\threeval(\varphi,j) = \top  \Rightarrow   [\pi_{j \cdots n}  \models_{3} \varphi] = \top)}_\text{Assumption step} \Rightarrow ( \underbrace{\underbrace{\threeval(\eventually \varphi,i) = \top}_\text{$\iff \threeval(\varphi,j) = \top$}}_\text{True by assumption}  \Rightarrow  \underbrace{\underbrace{[\pi_{i \cdots n}  \models_{3}  \eventually \varphi] = \top}_\text{$\iff [\pi_{j \cdots n}  \models_{3} \varphi] = \top$}}_\text{True by assumption})$$
$$(
\underbrace{\threeval(\varphi,j) = \bot  \Rightarrow  [\pi_{j \cdots n}  \models_{3} \varphi] = \bot)}_\text{Assumption step} \Rightarrow \underbrace{(\underbrace{\threeval(\eventually \varphi,i) = \bot}_\text{This is always false}  \Rightarrow  
\underbrace{[\pi_{i \cdots n}  \models_{3}  \eventually \varphi] = \bot)}_\text{This is always false}}_\text{This is true}$$

\item[Induction step for $\varphi_1 \until \varphi_2$] We can prove that:

\[ 
\left(
    \begin{array}{c}
    \exists j, i \leq j \leq |\pi| \mbox{ s.t.}\\
    \threeval(\varphi_1,j) = \top  \Rightarrow  [\pi_{j \cdots n} \models_{3} \varphi_1] = \top \mbox{ } \wedge \\
     \forall k, i \leq k < j \mbox{,}\\
    \threeval(\varphi_2,k) = \top  \Rightarrow  [\pi_{k \cdots n} \models_{3} \varphi_2] = \top
    \end{array}
  \right) \Rightarrow   \underbrace{\threeval(\varphi_1\; \until \; \varphi_2,i) = \top}_\text{True by  assumption}  \Rightarrow   \underbrace{[\pi_{i \cdots n} \models_{3} \varphi_1 \; \until \; \varphi_2] = \top}_\text{True by assumption}  
\]

\[ 
\left(
    \begin{array}{c}
    \forall j, i \leq j \leq |\pi| \mbox{ s.t.}\\
    \threeval(\varphi_1,j) = \top  \Rightarrow  [\pi_{j \cdots n} \models_{3} \varphi_1] = \top \mbox{ } \wedge \\
     \exists k, i \leq k < j \mbox{,}\\
    \threeval(\varphi_2,k) = \top  \Rightarrow  [\pi_{k \cdots n} \models_{3} \varphi_2] = \top
    \end{array}
  \right) \Rightarrow   \underbrace{\threeval(\varphi_1\; \until \; \varphi_2,i) = \bot}_\text{True by assumption}   \Rightarrow   \underbrace{[\pi_{i \cdots n} \models_{3} \varphi_1 \; \until \; \varphi_2] = \bot}_\text{True by  assumption}
\]

}

\end{description}

\end{proof} 

\subsection{Proof for Lemma~\ref{lemma:admissible}}
\label{proof:lemma:admissible}

Let $s,f \in \mathbb{N}_0 \cup \{\infty \}$.
We first define the following sets:
\begin{itemize}
\item $P^+_{i,  \pi} = \{ \; (s, -)\; | \; s \leq  |\pi| - i \ \}$
\item $P^-_{i,  \pi}  = \{ \; (-, f)\; | \; f \leq  |\pi| - i \ \}$
\item $P^{?}_{i,  \pi}  = \{ \; (s, f)\; | \; s, f > |\pi| - i \ \}$
\item  $P_{i,  \pi}  = P^+_{i,  \pi} \cup P^-_{i,  \pi}  \cup P^{?}_{i,  \pi}   $
\end{itemize}

The set $P^+_{i,  \pi}$ represents the set of all the possible pairs of the 
form $(a,-)$, the set $P^-_{i,  \pi}$ represents the set of all the possible 
pairs of the form $(-,a)$ while $P^{?}_{i,  \pi}$ represents the set of 
all the possible pairs of the form $(b_{1}, b_{2})$, $(b_{1},\infty)$, $(\infty, b_{2})$ or 
$(\infty, \infty)$, where $a \leq |\pi| - i$ and $b_{j} > |\pi| - i$ for $j \in \{1,2\}$. 

We now provide and prove the following proposition that will be used to
prove later Lemma~\ref{lemma:admissible}.

\begin{proposition}
\label{prop:endoftrace}
Let $\pi \in \Pi^{*}$ be a finite trace, $\phi$ an LTL formula and $i\in \mathbb{N}$ an index. 
Then we have that $\forall i  > |\pi|, \countsem(\phi,i) \in P^{?}_{i,  \pi}$.
\end{proposition} 

\begin{proof}
\begin{description}

{\small 
\item[Base case $\varphi ::= p$.] $\countsem(p,i)  =   (0,0)  \in P^{?}_{i,  \pi}   \text{  for } i > |\pi|  $
\item[Induction step $ \countsem(\varphi,i) \in P^{?}_{i,  \pi} \Rightarrow \countsem(\neg \varphi,i) \in P^{?}_{i,  \pi} $.] This is true because $\countsem(\varphi,i) \in P^{?}_{i,  \pi} \Rightarrow \swap(\countsem(\varphi,i))   \in  P^{?}_{i,  \pi}$. 
\item[Induction step $ \countsem(\varphi_1,i) \in P^{?}_{i,  \pi},  \countsem(\varphi_2,i) \in P^{?}_{i,  \pi}, \Rightarrow \countsem(\varphi_1 \vee \varphi_2,i) \in P^{?}_{i,  \pi} $.] 
This is true because if  $\countsem(\varphi_1,i) \in P^{?}_{i,  \pi},  \countsem(\varphi_2,i) \in P^{?}_{i,  \pi}$ then $\countsem(\varphi_1,i) = (s_1, f_1) \mbox{ and } \countsem(\varphi_2,i) =  (s_2, f_2)$
such that  $s_1,s_2, f_1, f_2 \in \mathbb{N}_0 \cup \{\infty \}$ and  $s_1,s_2, f_1, f_2 >  |\pi| - i $. Then $\countsem(\varphi_1 \vee \varphi_2,i) = (s_1, f_1) \minmax (s_2, f_2) = (\underbrace{\min(s_1,s_2)}_\text{$> |\pi| -i $},  \underbrace{\max(f_1,f_2)}_\text{$> |\pi| -i $} ) \in  P^{?}_{i,  \pi}$.
\item[Induction step $\forall i > |\pi|, \countsem(\varphi,i+1) \in P^{?}_{i+1,  \pi} \Rightarrow \countsem(\nextt \varphi,i) \in P^{?}_{i,  \pi} $.]  
If $\countsem(\varphi,i+1) \in P^{?}_{i+1,  \pi}$ then $\countsem(\varphi,i+1) = (s_1,f_1)$ 
such that  $s, f \in \mathbb{N}_0 \cup \{\infty \}$ and  $s,f >  |\pi| - i - 1$.  Then we have that  $ \countsem(\nextt \varphi, i) =  \underbrace{(s,f) \inc}_\text{$s \oplus 1, f \oplus 1 >  |\pi| - i $} \in  P^{?}_{i+1,  \pi}$.
\item[Induction step $\forall i > |\pi|, \countsem(\varphi_1,i),  \countsem(\varphi_2,i) \in P^{?}_{i,  \pi}, \Rightarrow \countsem(\varphi_1 \until \varphi_2,i) \in P^{?}_{i,  \pi} $.] 
 If  $\countsem(\varphi_1,i) \in P^{?}_{i,  \pi},  \countsem(\varphi_2,i) \in P^{?}_{i,  \pi}$ then $\countsem(\varphi_1,i) = (s_1, f_1) \mbox{ and } \countsem(\varphi_2,i) =  (s_2, f_2)$
such that  $s_1,s_2, f_1, f_2 \in \mathbb{N}_0 \cup \{\infty \}$ and  $s_1,s_2, f_1, f_2 >  |\pi| - i $. Using the definition of $\countsem$ when $i > |\pi|$, 
$\countsem(\varphi_1 \until \varphi_2,i) =  \underbrace{\Big( (s_1, f_1)  \minmax \Big( \underbrace{(s_2,f_2)  \maxmin  (\impossible, \infty)}_\text{$= (\max (s_2, \impossible), \min(f_2, \infty)) = (\impossible, f_2)$} \Big)  \Big)}_\text{$= (\min (s_1, \impossible ), \max (f_1,f_2) ) = (s_1,  \max (f_1,f_2) ) \in P^{?}_{i,  \pi} $} \in P^{?}_{i,  \pi}.$

\item[Induction step $ \countsem(\varphi,i) \in P^{?}_{i,  \pi} \Rightarrow \countsem(\eventually \varphi, i) \in P^{?}_{i,  \pi} $.]  If $\countsem(\varphi,i) \in P^{?}_{i,  \pi}$ then 
$\countsem(\varphi,i) = (s, f)$ such that  $s, f \in \mathbb{N}_0 \cup \{\infty \}$ and  $s, f >  |\pi| - i $. 
Then following the definition of $\countsem$ when $i > |\pi|$, 
$\countsem(\eventually \varphi, i) = \countsem(\varphi, i)\; \minmax (-, \infty) = (\underbrace{\min (s_1, \impossible)}_\text{$s_1 > |\pi| - i$}, \underbrace{\max(f_1, \infty)}_\text{$\infty > |\pi| - i$} ) \in P^{?}_{i,  \pi}.$ 
}
\end{description}
\end{proof}

In the following we now prove Lemma~\ref{lemma:admissible}. 
\begin{proof} 

We now need to prove the closure of $P_{i,  \pi}$ under $\countsem(\phi,i)$
inductively on the structure of the LTL formula by considering all 
the possible cases.

\begin{description}

{\small \item[Base case $\varphi ::= p$] 

$\begin{array}{lcl}
\countsem(p,i) & = & \begin{cases}
 (0, -) \in P^+_{i,  \pi}  & \text{if } i \leq |\pi| \wedge p \in \pi_{i} \\
 (-, 0)  \in P^-_{i,  \pi}  & \text{if } i \leq |\pi| \wedge p \not \in \pi_{i} \\
 (0,0)  \in P^{?}_{i,  \pi}  & \text{if } i > |\pi|  \\
\end{cases} 
\end{array}
$
}
\\
{\small \item[Induction step $ \countsem(\varphi,i)  \in P_{i,  \pi} \Rightarrow  \countsem(\neg \varphi,i)  \in P_{i,  \pi}  $]  We have  three cases:
\begin{description}
\item [($\countsem(\varphi,i)  \in P^+_{i,  \pi} $)] $    \countsem(\neg \varphi,i)  =    \swap(\countsem(\varphi,i))   \in P^-_{i,  \pi} $ 
\item [($\countsem(\varphi,i)  \in P^-_{i,  \pi} $)] $ \countsem(\neg \varphi,i)  =    \swap(\countsem(\varphi,i))   \in P^+_{i,  \pi} $ 
\item [($\countsem(\varphi,i)  \in P^{?}_{i,  \pi} $)]$ \countsem(\neg \varphi,i)  =    \swap(\countsem(\varphi,i))   \in  P^{?}_{i,  \pi}   $\\
\end{description}}
{\small \item[Induction step  $A = \countsem(\varphi_1,i)  \in P_{i,  \pi},  B = \countsem(\varphi_2,i)  \in P_{i,  \pi} \Rightarrow \countsem(\varphi_1  \vee \varphi_2,i)  \in P_{i,  \pi}   $] 
 We need to consider the following cases:
\begin{description}
\item [($A \in P^+_{i,  \pi} $, $B \in P^+_{i,  \pi} $)] $A=(s_1,-), B=(s_2,-), \underbrace{A \minmax B= (\min(s_1,s_2),-)}_\text{$\in    P^+_{i,  \pi}  \subset P_{i,  \pi}  $}$
\item [($A \in P^-_{i,  \pi} $, $B \in P^+_{i,  \pi} $)] $ A \minmax B = B \in P^+_{i,  \pi}  \subset P_{i,  \pi} $
\item [($A \in P^{?}_{i,  \pi}  $, $B \in P^+_{i,  \pi} $)]  $A=(s_1,f_1), B=(s_2,-),\; s_1 > |\pi| - i, s_2 \leq |\pi| - i  \Rightarrow s_2 < s_1$  \\

$A \minmax B= (\min(s_1,s_2),\max(f_1,-)) = (s_2, -) = B \in    P^+_{i,  \pi}  \subset P_{i,  \pi} $\\

\item [($A \in P^+_{i,  \pi}, B \in  P^-_{i,  \pi} $ )]  Since $\minmax$  is commutative see the case  ($A \in  P^{-}_{i,  \pi}$, $B \in P^+_{i,  \pi}$) \\

\item [($A \in P^-_{i,  \pi}, B \in  P^-_{i,  \pi}$)]  $A=(-,f_1)$, $B=(-,f_2), \underbrace{A \minmax B = \{(-,\max(f_1,f_2))\}}_\text{$\in  P^-_{i,  \pi} \subset P_{i,  \pi} $}$

\item [($A \in P^{?}_{i,  \pi}, B \in  P^-_{i,  \pi}$)] $A=(s_1,f_1), B=(-,f_2),\;  f_1 > |\pi| - i, f_2 \leq |\pi| - i  \Rightarrow f_1 > f_2$\\

$A \minmax B= (\min(s_1,-),\max(f_1,f_2)) = (s_1, f_1) = A \in     P^{?}_{i,  \pi} \subset P_{i,  \pi} $\\

\item [($A \in P^+_{i,  \pi}, B \in  P^{?}_{i,  \pi} $)] Since $\minmax$  is commutative see the case  ($A \in  P^{?}_{i,  \pi}$, $B \in P^+_{i,  \pi}$).
\item [($A \in P^-_{i,  \pi}, B \in  P^{?}_{i,  \pi} $)] Since $\minmax$ is commutative see the case  ($A \in  P^{?}_{i,  \pi}$, $B \in P^-_{i,  \pi}$).
\item [($A \in  P^{?}_{i+1,  \pi}, B \in P^{?}_{i,  \pi}$)] $A=(s_1,f_1), B=(s_2,f_2),\;  s_1,f_1,s_2,f_2 > |\pi| - i$\\

$A \minmax B= (\min(s_1,s_2),\max(f_1,f_2))  \in  P^{?}_{i,  \pi} \subset P_{i,  \pi} $\\

\end{description}}

{\small  \item[Induction step $A = \countsem(\varphi,i+1) \in  P_{i+1,  \pi} \Rightarrow   \countsem(\nextt \varphi,i) \in  P_{i,  \pi} $] We have  three cases:
\begin{description}
\item [($A \in P^+_{i+1, \pi}$)] $ \underbrace{\countsem(\nextt \varphi,i) = A \inc = (s_1+1, -)}_\text{$s_1 \leq |\pi| - i -1 \Rightarrow s_1 + 1 \leq |\pi| - i$} \in  P^+_{i,  \pi}$
\item [($A \in P^-_{i+1,  \pi}$)]  $ \underbrace{\countsem(\nextt \varphi,i) = A \inc = (-,f_1+1)}_\text{$f_1 \leq |\pi| - i -1 \Rightarrow f_1 + 1 \leq |\pi| - i$} \in  P^-_{i,  \pi}$
\item [($A \in P^{?}_{i+1,  \pi}$)]  $ \underbrace{\countsem(\nextt \varphi,i) = A \inc = (s_1 \oplus 1, f_1 \oplus 1)  \in   P^{?}_{i,  \pi}}_\text{$s_1 > |\pi| - i -1 \Rightarrow s_1 \oplus 1 > |\pi| - i, f_1 > |\pi| - i -1 \Rightarrow f_1 \oplus 1 > |\pi| - i$}$
\\\\
\end{description}}

{\small \item[Induction step  $A = \countsem(\varphi,j) \in  P_{j,  \pi} \Rightarrow   \countsem(\eventually \varphi,i) \in  P_{i,  \pi} $.] 
$$ \mbox{if } i > |\pi| \Rightarrow A \in   P^{?}_{i,  \pi} \mbox{ } \Rightarrow \mbox{ } \countsem(\eventually \varphi,i) \in P^{?}_{i,  \pi} \subset   P_{i,  \pi}  \mbox{ (See Prop.~\ref{prop:endoftrace})}$$ 
$$\mbox{if } i \leq |\pi|  \Rightarrow \countsem( \eventually \varphi, i) = \countsem(\phi,  i)  \minmax  \countsem(\nextt \! (\eventually \varphi), i ) $$

$$ \countsem( \eventually \varphi, i+1)  \in  P_{i+1,  \pi}  \Rightarrow  \countsem(\phi,  i)  \minmax  \countsem(\nextt \! (\eventually \varphi), i ) \in P_{i,  \pi}  $$
$$\mbox{and we proved that at least when } i+1 > |\pi|,  \countsem(\eventually \varphi,i+1) \in   P^{?}_{i+1,  \pi} \subset  P_{i+1,  \pi}   $$.
}
{\small \item[Induction step  $A = \countsem(\varphi_1,i)  \in P_{i,  \pi},  B = \countsem(\varphi_2,i)  \in P_{i,  \pi} \Rightarrow \countsem(\varphi_1  \until \varphi_2,i)  \in P_{i,  \pi} $.] 
$$\mbox{if } i > |\pi| \Rightarrow A,B \in  A \in P^{?}_{i,  \pi} \Rightarrow \mbox{ } \countsem(\varphi_1  \until \varphi_2,i)  \in P^{?}_{i,  \pi}  \subset P_{i,  \pi} \mbox{ (See Prop.~\ref{prop:endoftrace})}$$  

$$ \mbox{if } i \leq |\pi|  \Rightarrow \countsem(\varphi_1  \until \varphi_2,i) =  A \minmax (B \maxmin ( \countsem(\nextt (\varphi_1  \until \varphi_2),i) )) $$

$$ \countsem(\varphi_1  \until \varphi_2,i+1) \in  P_{i+1,  \pi}  \Rightarrow A \minmax (B \maxmin ( \countsem(\nextt (\varphi_1  \until \varphi_2),i) )) \in P_{i,  \pi}  $$
$$\mbox{and we  proved that at least when } i+1 > |\pi|,   \countsem(\varphi_1  \until \varphi_2,i+1) \in  P^{?}_{i+1,  \pi} \subset  P_{i+1,  \pi}  $$.
}
\end{description}
\end{proof}

\section{Examples}
\label{sec:examples}


\paragraph{Evaluation of the Next Operator:}
In Table~\ref{sec:finiteLTL:sec:Examples:X:methodEval} we illustrate the evaluation of the $\nextt$ operator nested in an $\eventually$ property and nested  in a $\always$ property.

Our approach focuses on observed past behavior and predicts evaluations of subformulas when possible. The prediction on $\nextt g$  is necessary to draw a conclusion on the eventually, respectively globally, property being violated, respectively satisfied.
For the trace in Table~\ref{sec:finiteLTL:sec:Examples:X:methodEval} (a) our approach results in the expected presumably false verdict, because we have always observed $\nextt g$  being violated and we do not expect it to be satisfied. For the trace in Table~\ref{sec:finiteLTL:sec:Examples:X:methodEval} (b) our approach results in the expected presumably true verdict, because we have always observed $\nextt g$ being satisfied and we do not expect it to be violated.

 \vspace{-2ex}
 
 \begin{table}
\centering
\caption{Evaluation of the $\nextt$ operator nested in an $\eventually$ and a $\always$ property.}
\label{sec:finiteLTL:sec:Examples:X:methodEval}
\subfloat[$\eventually   \nextt g$.] {\tiny
\begin{tabular}  { c | c c c c c c c c r}
 i  &  1 & 2 & 3 & 4  & EOT \\ 
\toprule
$g$& $ - $  & $ - $  & $ - $  & $ - $  \\ 
\midrule
$d_\pi(g,i)$   & $(-,0)$  & $(-,0)$  & $(-,0)$  & $(-,0)$  & $(0,0)$   &\multirow{2}{*}{(1)}\\ \\  
$e_\pi(g,i)$  & $ \bot $ & $ \bot $ & $ \bot $ & $ \bot $ &  $ \bot_P $ \\
\midrule
$d_\pi(\nextt g,i)$ & $(-,1)$  & $(-,1)$  & $(-,1)$  & $(1,1)$  & $(1,1)$   &\multirow{2}{*}{(2)}\\ \\  
$e_\pi(\nextt g,i)$   & $ \bot $ & $ \bot $ & $ \bot $  & $ \bot_P $ & $ \bot_P $ \\
\midrule
$d_\pi(\eventually   \nextt g,i)$  & $(4,\infty)$  & $(3,\infty)$  & $(2,\infty)$  & $(1,\infty)$  & $(1,\infty)$   &\multirow{2}{*}{(3)}\\ \\  
$e_\pi(\eventually   \nextt g,i)$ & $ \bot_P $ & $ \bot_P $ & $ \bot_P $ & $ \bot_P $ & $ \bot_P $ \\
\bottomrule
\end{tabular}
}
\qquad
\subfloat[$\always   \nextt g$.]{%
 {\tiny
\begin{tabular}  { c | c c c c c c c c r}
 i   & 1 & 2 & 3 & 4 & EOT \\ 
\toprule
$\pi_g$   & $ \top $  & $ \top $  & $ \top $  & $ \top $ \\ 
\midrule
$d_\pi(g,i)$& $(0,-)$  & $(0,-)$  & $(0,-)$  & $(0,-)$  & $(0,0)$  &\multirow{2}{*}{(1)}\\  \\  
$e_\pi(g,i)$  & $ \top $ & $ \top $ & $ \top $ & $ \top $ & $ \top_P $ &\\
\midrule
$d_\pi(\nextt g,i)$   & $(1,-)$  & $(1,-)$  & $(1,-)$  & $(1,1)$  & $(1,1)$  &\multirow{2}{*}{(2)}\\  \\  
$e_\pi(\nextt g,i)$ & $ \top $ & $ \top $ & $ \top $ & $ \top_P $ & $ \top_P $ \\
\midrule
$d_\pi(\always   \nextt g,i)$ & $(\infty,4)$  & $(\infty,3)$  & $(\infty,2)$  & $(\infty,1)$  & $(\infty,1)$  &\multirow{2}{*}{(3)}\\  \\  
$e_\pi(\always   \nextt g,i)$ & $ \top_P $ & $ \top_P $ & $ \top_P $ & $ \top_P $ & $ \top_P $ \\
\bottomrule
\end{tabular}
}}
\end{table}

\paragraph{Request/Acknowledge Properties:}
As a running example we have already illustrated the evaluation of trace $\pi_1$ from the motivation with the property
$$
\always(r \rightarrow \eventually g)\text{.}
$$
We now also evaluate the second trace from the motivation. 
In Table~\ref{sec:finiteLTL:sec:Examples:tab:response-pi2} we present the evaluation. While for many positions (like $i=5$) the signal $r$ dominates (because it is false and, thus, the implication is trivially satisfied) this is not the case for position $i=4$.  At this position the implication is not yet satisfied within the trace and, thus, can be  at earliest satisfied in 4 steps by extending the trace with $g=\true$ at $i=8$. However, the longest observed witness for satisfaction of the implication is at $i=1$ and requires two additional steps. As we've never observed a witness that requires at least 4 additional steps for a satisfaction, the suffix  at $i=4$ is concluded to be presumably false. Hence, the globally property is expected to be violated and we conclude that this trace is going to presumably violate the given property.

\begin{table}[htb]
\centering
\caption{Trace $\pi_2$ from the motivation.}
   {\tiny
\begin{tabular}{r cccccccc r}
\toprule
	i  &  1	& 2	& 3	& 4	& 5	& 6	& 7	  & EOT \\

\midrule
$r$    & $\top$ & $-$ & $-$ & $\top$ & $-$ & $-$ & $-$ & \\
$g$     &  $-$ & $-$ & $\top$ & $-$ & $-$ & $-$ & $-$ &  \\
\bottomrule
$\countsem(r,i)$ & $(0,-)$  & $(-,0)$  & $(-,0)$  & $(0,-)$  & $(-,0)$  & $(-,0)$  & $(-,0)$  & $(0,0)$  &\multirow{2}{*}{(1)}\\ 
$\predictive(r,i) $ & $ \top $ & $ \bot $ & $ \bot $ & $ \top $ & $ \bot $ & $ \bot $ & $ \bot $ & $ ? $ \\
\midrule
$\countsem(\neg r,i)$ & $(-,0)$  & $(0,-)$  & $(0,-)$  & $(-,0)$  & $(0,-)$  & $(0,-)$  & $(0,-)$  & $(0,0)$   &\multirow{2}{*}{(2)}\\ 
$\predictive(\neg r,i)$ & $ \bot $ & $ \top $ & $ \top $ & $ \bot $ & $ \top $ & $ \top $ & $ \top $ & $ ? $ \\
\midrule
$\countsem(g,i)$ & $(-,0)$  & $(-,0)$  & $(0,-)$  & $(-,0)$  & $(-,0)$  & $(-,0)$  & $(-,0)$  & $(0,0)$   &\multirow{2}{*}{(3)}\\ 
$\predictive(g,i)$  & $ \bot $ & $ \bot $ & $ \top $ & $ \bot $ & $ \bot $ & $ \bot $ & $ \bot $ & $ ? $ \\
\midrule
$\countsem(\eventually g,i)$ & $(2,-)$  & $(1,-)$  & $(0,-)$  & $(4,\infty)$  & $(3,\infty)$  & $(2,\infty)$  & $(1,\infty)$  & $(0,\infty)$  &\multirow{2}{*}{(4)}\\ 
$\predictive(\eventually g,i)$  & $ \top $ & $ \top $ & $ \top $ & $ \bot_P $ & $ \bot_P $ & $ \top_P $ & $ \top_P $ & $ \top_P $ \\
\midrule
$\countsem(r \rightarrow \eventually g,i)$ & $(2,-)$  & $(0,-)$  & $(0,-)$  & $(4,\infty)$  & $(0,-)$  & $(0,-)$  & $(0,-)$  & $(0,\infty)$   &\multirow{2}{*}{(5)}\\ 
$\predictive(r \rightarrow \eventually g,i)$ & $ \top $ & $ \top $ & $ \top $ & $ \bot_P $ & $ \top $ & $ \top $ & $ \top $ & $ \top_P $ \\
\midrule
$\countsem(\always(r \rightarrow \eventually g,i))$ & $(\infty,\infty)$  & $(\infty,\infty)$  & $(\infty,\infty)$  & $(\infty,\infty)$  & $(\infty,\infty)$  & $(\infty,\infty)$  & $(\infty,\infty)$  & $(\infty,\infty)$  &\multirow{2}{*}{(6)}\\ 
$\predictive(\always(r \rightarrow \eventually g,i))$ & $ \bot_P $ & $ \bot_P $ & $ \bot_P $ & $ \bot_P $ & $ \top_P $ & $ \top_P $ & $ \top_P $ & $ \top_P $\\
\bottomrule

\end{tabular}
\label{sec:finiteLTL:sec:Examples:tab:response-pi2}
}
\end{table}

Next we illustrate in Table~\ref{sec:finiteLTL:sec:Examples:tab:NeedForPrediction} why predictions on the different levels of subformulas are necessary. 
Note that the prediction for the property $\eventually g$ at Position~5 is $\ptrue$, because there exists a witness in the past (at Position~1) that required the same amount of additional steps for satisfaction. when evaluating the property $r \rightarrow \eventually g$, the prediction for the same Position becomes $\pfalse$, because now the longest witness (at Position~2) only requires one additional step, which is shorter than the required two additional steps (at Position~5).
This is, because the signal $g$ is related to the signal $r$, and at Position~1 the truth value of signal $r$ dominates. Human intuition supports this evaluation. 
While evaluating only $\eventually g$ allows the observer to conclude that it always takes two additional steps to observe the grant, this is not the case when evaluating $r \rightarrow \eventually g$. For this property, the signal $g$ is only relevant whenever a request $r$ is observed and then the grant $g$ is observed in one additional step.

\begin{table}[tb]
\centering
  \caption{Need for prediction of individual subformulas.} 
     {\scriptsize
  \begin{tabular}{ c  c c c c c c c r}
\toprule
 i  & 1 & 2 & 3 & 4 & 5 & 6 & EOT \\ 
 \midrule
$r$ & $ - $  & $ \top $  & $ - $  & $ - $  & $ \top $  & $ - $ \\ 
$g$ & $ - $  & $ - $  & $ \top $  & $ - $  & $ - $  & $ - $ \\ 
\bottomrule
$\countsem(r,i)$ & $(-,0)$  & $(0,-)$  & $(-,0)$  & $(-,0)$  & $(0,-)$  & $(-,0)$  & $(0,0)$ &\multirow{2}{*}{(1)}\\ 
$\predictive(r,i)$ & $ \bot $ & $ \top $ & $ \bot $ & $ \bot $ & $ \top $ & $ \bot $ & $ ? $ \\
\midrule
$\countsem(g,i)$ & $(-,0)$  & $(-,0)$  & $(0,-)$  & $(-,0)$  & $(-,0)$  & $(-,0)$  & $(0,0)$ &\multirow{2}{*}{(2)}\\ 
$\predictive(g,i)$ & $ \bot $ & $ \bot $ & $ \top $ & $ \bot $ & $ \bot $ & $ \bot $ & $ ? $ \\
\midrule
$\countsem(\eventually g,i)$ & $(2,-)$  & $(1,-)$  & $(0,-)$  & $(3,\infty)$  & $(2,\infty)$  & $(1,\infty)$  & $(0,\infty)$ &\multirow{2}{*}{(3)}\\ 
$\predictive(\eventually g,i)$ & $ \top $ & $ \top $ & $ \top $ & $ \pfalse $ & $ \ptrue $ & $ \ptrue $ & $ \ptrue $ \\
\midrule
$\countsem(r \rightarrow \eventually g,i)$ & $(0,-)$  & $(1,-)$  & $(0,-)$  & $(0,-)$  & $(2,\infty)$  & $(0,-)$  & $(0,\infty)$ &\multirow{2}{*}{(4)}\\  
$\predictive(r \rightarrow \eventually g,i)$ & $ \top $ & $ \top $ & $ \top $ & $ \top $ & $ \pfalse $ & $ \top $ & $ \ptrue $ \\
\midrule
$\countsem(\always(r \rightarrow \eventually g),i)$ & $(\infty,\infty)$  & $(\infty,\infty)$  & $(\infty,\infty)$  & $(\infty,\infty)$  & $(\infty,\infty)$  & $(\infty,\infty)$  & $(\infty,\infty)$ &\multirow{2}{*}{(5)}\\ 
$\predictive(\always(r \rightarrow \eventually g),i)$ & $ \pfalse $ & $ \pfalse $ & $ \pfalse $ & $ \pfalse $ & $ \pfalse $ & $ \ptrue $ & $ \ptrue $ \\
\midrule
  \end{tabular}
\label{sec:finiteLTL:sec:Examples:tab:NeedForPrediction}
}
\end{table}

In another request/acknowledge example we analyze the property
$$
\always ( r_1 \rightarrow \eventually g_1) \land \always ( r_2 \rightarrow \eventually g_2)
$$
with $r_1$ being triggered at even time steps, $r_2$ being triggered at odd time steps, and both requests being always granted after exactly one time step. No matter where you cut the trace there is always one request not yet granted (Table~\ref{sec:finiteLTL:sec:Examples:2GrFg:inputseq} illustrates an example trace).

The two request/grant properties are conjunct on the highest level of the formula. Our approach computes truth values for every subformula, i.e., computes independent predictions for both request/grant properties which is in both cases $\top_P$. On the highest level (no predictions are possible anymore at this level, because all computed pairs are of the form $(\infty,\infty)$) the computed truth values for the two request/grant properties are conjunct and result in the expected verdict  presumably true.

\begin{table}[tb]
            \centering
  \caption{Trace of a system claiming to implement $ \always (\neg r_1 \lor \eventually g_1) \land \always (\neg r_2 \lor \eventually g_2)$.}
\label{sec:finiteLTL:sec:Examples:2GrFg:inputseq}
   {\scriptsize
  \begin{tabular}{ c  c c c c c c c c c c c c c }
  \toprule
   &  1 & 2 & 3 & 4 & 5 & 6 & 7 & 8 & 9 & 10 & 11 & 12 & 13\\
  \midrule
 $r_1 $  & $ \top $  & $ - $       & $ \top $  & $ - $       & $ \top $  & $ - $       & $ \top$  & $ - $      & $ \top $  & $ - $       & $\top $ & $ - $     & $ \top $\\ 
 $g_1$ & $ - $      & $ \top $  & $ - $        & $ \top $  & $ - $       & $ \top $  & $ - $      & $ \top$  & $ - $       & $ \top $  & $ - $     & $\top $ & $ - $ \\ 
 $r_2$ & $ - $  & $ \top $  & $ - $  & $ \top $  & $ - $  & $ \top $  & $ - $  & $ \top$  & $ - $  & $ \top $  & $ - $  & $\top $ & $ - $ \\ 
 $g_2$ & $ - $  & $ - $  & $ \top $  & $ - $  & $ \top $  & $ - $  & $ \top$  & $ - $  & $ \top $  & $ - $  & $\top $ & $ - $ & $ \top $\\ 
\bottomrule
  \end{tabular}
}
 \end{table}

\paragraph{Evaluation of the Until Operator:} To illustrate our approach on a specification that contains an until operator, we consider the property
$$
\always ( (\nextt a) \until \nextt   \nextt b).
$$ 

\noindent Table~\ref{sec:finiteLTL:sec:Examples:tab:GXaUXXb:methodEval} shows an example trace and the associated evaluation. 
The longest observed witness for satisfaction of the until property starts at position~1 and requires six additional time steps. In positions~1,~2,~3 and~4 the subformula $\nextt a$ holds, until in position~5 the subformula $\nextt \nextt b$ holds.
The suffix of the trace from position~6 can be satisfied at earliest after 3 time steps by an extension of the trace with  $b=\top$ at $i=9$. As the suffix is shorter than 
the longest observed witness for satisfaction and we have not observed any violation, this inconclusive suffix is predicted to be 
presumably true. The same applies for the suffixes starting at $i=7$ and $i=8$. Thus, we neither observe nor expect a 
violation of the globally property. Hence, the property evaluates to $\top_P$ with respect to the given trace.

\begin{table}[tb]
\center
  \caption{Evaluation of $\always((\nextt a) \until \nextt   \nextt b)$.}   
   \label{sec:finiteLTL:sec:Examples:tab:GXaUXXb:methodEval}
{\tiny
  \begin{tabular}{ r  c c c c c c c c c c c r}
  \toprule
 i  & 1 & 2 & 3 & 4 & 5 & 6 & 7 & 8 & EOT \\ 
\midrule
$a$  & $-$ & $\top$  & $\top$  & $\top$  & $\top$ & $-$  & $\top$ & $\top$ \\
$b$ & $-$ & $\top$ & $-$ & $-$ & $-$ & $-$ & $\top$ & $-$ \\
\bottomrule
$\countsem(a,i)$ & $(-,0)$  & $(0,-)$  & $(0,-)$  & $(0,-)$  & $(0,-)$  & $(-,0)$  & $(0,-)$  & $(0,-)$  & $(0,0)$   &\multirow{2}{*}{(1)}\\ 
$\predictive(a,i)$   & $ \bot $ & $ \top $ & $ \top $ & $ \top $ & $ \top $ & $ \bot $ & $ \top $ & $ \top $ & $ ? $ \\
\midrule
$\countsem(\nextt a,i)$  & $(1,-)$  & $(1,-)$  & $(1,-)$  & $(1,-)$  & $(-,1)$  & $(1,-)$  & $(1,-)$  & $(1,1)$  & $(1,1)$   &\multirow{2}{*}{(2)}\\ 
$\predictive(\nextt a,i)$  & $ \top $ & $ \top $ & $ \top $ & $ \top $ & $ \bot $ & $ \top $ & $ \top $ & $ ? $ & $ ? $ \\
\midrule
$\countsem(b,i)$   & $(-,0)$  & $(0,-)$  & $(-,0)$  & $(-,0)$  & $(-,0)$  & $(-,0)$  & $(0,-)$  & $(-,0)$  & $(0,0)$  &\multirow{2}{*}{(3)}\\ 
$\predictive(b,i)$& $ \bot $ & $ \top $ & $ \bot $ & $ \bot $ & $ \bot $ & $ \bot $ & $ \top $ & $ \bot $ & $ ? $ \\
\midrule
$\countsem(\nextt b,i)$ & $(1,-)$  & $(-,1)$  & $(-,1)$  & $(-,1)$  & $(-,1)$  & $(1,-)$  & $(-,1)$  & $(1,1)$  & $(1,1)$   &\multirow{2}{*}{(4)}\\ 
$\predictive(\nextt b,i)$& $ \top $ & $ \bot $ & $ \bot $ & $ \bot $ & $ \bot $ & $ \top $ & $ \bot $ & $ ? $ & $ ? $ \\
\midrule
$\countsem(\nextt   \nextt b,i)$   & $(-,2)$  & $(-,2)$  & $(-,2)$  & $(-,2)$  & $(2,-)$  & $(-,2)$  & $(2,2)$  & $(2,2)$  & $(2,2)$   &\multirow{2}{*}{(5)}\\ 
$\predictive(\nextt   \nextt b,i)$  & $ \bot $ & $ \bot $ & $ \bot $ & $ \bot $ & $ \top $ & $ \bot $ & $ ? $ & $ ? $ & $ ? $ \\
\midrule
$\countsem(\nextt a \until \nextt   \nextt b,i)$  & $(6,-)$  & $(5,-)$  & $(4,-)$  & $(3,-)$  & $(2,-)$  & $(3,4)$  & $(2,3)$  & $(2,2)$  & $(2,2)$  &\multirow{2}{*}{(6)}\\ 
$\predictive(\nextt a \until \nextt   \nextt b,i)$ & $ \top $ & $ \top $ & $ \top $ & $ \top $ & $ \top $ & $ \top_P $ & $ \top_P $ & $ \top_P $ & $ \top_P $ \\
\midrule
$\countsem(\always(\nextt a \until \nextt   \nextt b),i)$ & $(\infty,9)$  & $(\infty,8)$  & $(\infty,7)$  & $(\infty,6)$  & $(\infty,5)$  & $(\infty,4)$  & $(\infty,3)$  & $(\infty,2)$  & $(\infty,2)$  &\multirow{2}{*}{(7)}\\ 
$\predictive(\always(\nextt a \until \nextt   \nextt b),i)$ & $ \top_P $ & $ \top_P $ & $ \top_P $ & $ \top_P $ & $ \top_P $ & $ \top_P $ & $ \top_P $ & $ \top_P $ & $ \top_P $ \\
\bottomrule
\end{tabular}
  }
\end{table}


\begin{table}[tb]
            \centering
  \caption{Traces of two systems that claim to implement $ \eventually   \always a \lor \eventually   \always \neg a$.}
\label{sec:finiteLTL:sec:Examples:FGaFGna:inputseq}
  {\scriptsize
  \begin{tabular}{ c  c c c c c c c c c c c c c}
  \toprule
   &  1 & 2 & 3 & 4 & 5 & 6 & 7 & 8 & 9 & 10 & 11 & 12 & 13\\
  \midrule
${\pi}_1$: a & $ \top $  & $ \top $  & $ - $  & $ - $  & $ \top $  & $ \top $  & $ - $  & $ - $  & $ \top $  & $ \top $  & $ -$ & $ - $ & $ \top $\\ 
\midrule
${\pi}_2$: a & $ \top $  & $ \top $  & $ - $  & $ - $  & $ \top $  & $ \top $  & $ - $  & $ - $  & $ \top $  & $ \top $  & $ \top$ & $ \top $ & $ \top $\\ 
\bottomrule
  \end{tabular}

}
\end{table}

\paragraph{Stabilization Properties:}
Consider the  property
$$
\eventually   \always a \lor  \eventually   \always \neg a
$$ that states that eventually the truth value of $a$ has to stabilize.

We analyze the traces presented in Table~\ref{sec:finiteLTL:sec:Examples:FGaFGna:inputseq}. While in trace $\pi_1$ the system seems to flip the truth value of $a$ always after time time steps, in trace $\pi_2$ the truth value of $a$ seems to remain stable from $i=9$ onwards.
Applying our approach, the first sequence ($\pi_1$) evaluates to presumably false because the suffix with one time $a=\top$ is shorter than a previous observed sequence of $a$s being stable (e.g. at position $i=1$ the truth value of $a$ was stable for two time steps). In the second sequence, the suffix with five times $a=\top$  is longer than any previous sequence of $a$s being stable and, thus, our approach evaluates this trace to presumably true.

These two examples also illustrate the importance of having a trace not truncated too early. Imagine cutting the trace at $i=5$ or $i=9$, then both traces evaluate to presumably false with respect to previously observed behavior, because we miss the observation of the long stable suffix.


 \begin{table}  [tb]
            \centering
  \caption{Evaluation of $\always(\eventually a \lor\eventually b)$.} 
   \label{sec:finiteLTL:sec:Examples:GFaFb1}  
  {\scriptsize
  \begin{tabular}{ c  c c c c c c c}
\toprule
 i  & 1 & 2 & 3 & 4 & 5 & 6 & EOT \\ 
 \midrule
a & $ \top $  & $ \top $  & $ \top $  & $ \top $  & $ - $  & $ - $ \\ 
b & $ \top $  & $ - $  & $ \top $  & $ - $  & $ \top $  & $ - $ \\ 
\bottomrule
$\countsem(a,i)$ & $(0,-)$  & $(0,-)$  & $(0,-)$  & $(0,-)$  & $(-,0)$  & $(-,0)$  & $(0,0)$  \\ \nopagebreak 
$\predictive(a,i)$ & $ \top $ & $ \top $ & $ \top $ & $ \top $ & $ \bot $ & $ \bot $ & $ ? $ \\
\midrule
$\countsem(\eventually a,i)$ & $(0,-)$  & $(0,-)$  & $(0,-)$  & $(0,-)$  & $(2,\infty)$  & $(1,\infty)$  & $(0,\infty)$  \\ \nopagebreak 
$\predictive(\eventually a,i)$ & $ \top $ & $ \top $ & $ \top $ & $ \top $ & $ \bot_P $ & $ \bot_P $ & $ \top_P $ \\
\midrule
$\countsem(b,i)$ & $(0,-)$  & $(-,0)$  & $(0,-)$  & $(-,0)$  & $(0,-)$  & $(-,0)$  & $(0,0)$  \\ \nopagebreak 
$\predictive(b,i)$ & $ \top $ & $ \bot $ & $ \top $ & $ \bot $ & $ \top $ & $ \bot $ & $ ? $ \\
\midrule
$\countsem(\eventually b,i)$ & $(0,-)$  & $(1,-)$  & $(0,-)$  & $(1,-)$  & $(0,-)$  & $(1,\infty)$  & $(0,\infty)$  \\ \nopagebreak 
$\predictive(\eventually b,i)$ & $ \top $ & $ \top $ & $ \top $ & $ \top $ & $ \top $ & $ \top_P $ & $ \top_P $ \\
\midrule
$\countsem(\eventually a \lor \eventually b,i)$ & $(0,-)$  & $(0,-)$  & $(0,-)$  & $(0,-)$  & $(0,-)$  & $(1,\infty)$  & $(0,\infty)$  \\ \nopagebreak 
$\predictive(\eventually a \lor\eventually b,i)$ & $ \top $ & $ \top $ & $ \top $ & $ \top $ & $ \top $ & $ \bot_P $ & $ \top_P $ \\
\midrule
$\countsem(\always(\eventually a \lor\eventually b),i)$ & $(\infty,\infty)$  & $(\infty,\infty)$  & $(\infty,\infty)$  & $(\infty,\infty)$  & $(\infty,\infty)$  & $(\infty,\infty)$  & $(\infty,\infty)$  \\ \nopagebreak 
$\predictive(\always(\eventually a \lor\eventually b),i)$ & $ \bot_P $ & $ \bot_P $ & $ \bot_P $ & $ \bot_P $ & $ \bot_P $ & $ \bot_P $ & $ \top_P $ \\
\bottomrule
  \end{tabular}
  }
\end{table}

 \begin{table}  [tb]
            \centering
  \caption{Evaluation of $\always \eventually a \lor \always \eventually b$.} 
   \label{sec:finiteLTL:sec:Examples:GFaFb2}  
  {\scriptsize
  \begin{tabular}{ c  c c c c c c c}
\toprule
 i  & 1 & 2 & 3 & 4 & 5 & 6 & EOT \\ 
 \midrule
a & $ \top $  & $ \top $  & $ \top $  & $ \top $  & $ - $  & $ - $ \\ 
b & $ \top $  & $ - $  & $ \top $  & $ - $  & $ \top $  & $ - $ \\ 
\bottomrule
$\countsem(a,i)$ & $(0,-)$  & $(0,-)$  & $(0,-)$  & $(0,-)$  & $(-,0)$  & $(-,0)$  & $(0,0)$  \\ \nopagebreak 
$\predictive(a,i)$ & $ \top $ & $ \top $ & $ \top $ & $ \top $ & $ \bot $ & $ \bot $ & $ ? $ \\
\midrule
$\countsem(\eventually a,i)$ & $(0,-)$  & $(0,-)$  & $(0,-)$  & $(0,-)$  & $(2,\infty)$  & $(1,\infty)$  & $(0,\infty)$  \\ \nopagebreak 
$\predictive(\eventually a,i)$ & $ \top $ & $ \top $ & $ \top $ & $ \top $ & $ \bot_P $ & $ \bot_P $ & $ \top_P $ \\
\midrule
$\countsem(\always \eventually a,i)$ & $(\infty,\infty)$  & $(\infty,\infty)$  & $(\infty,\infty)$  & $(\infty,\infty)$  & $(\infty,\infty)$  & $(\infty,\infty)$  & $(\infty,\infty)$  \\ \nopagebreak 
$\predictive(\always \eventually a,i)$ & $ \bot_P $ & $ \bot_P $ & $ \bot_P $ & $ \bot_P $ & $ \bot_P $ & $ \bot_P $ & $ \top_P $ \\
\midrule
$\countsem(b,i)$ & $(0,-)$  & $(-,0)$  & $(0,-)$  & $(-,0)$  & $(0,-)$  & $(-,0)$  & $(0,0)$  \\ \nopagebreak 
$\predictive(b,i)$ & $ \top $ & $ \bot $ & $ \top $ & $ \bot $ & $ \top $ & $ \bot $ & $ ? $ \\
\midrule
$\countsem(\eventually b,i)$ & $(0,-)$  & $(1,-)$  & $(0,-)$  & $(1,-)$  & $(0,-)$  & $(1,\infty)$  & $(0,\infty)$  \\ \nopagebreak 
$\predictive(\eventually b,i)$ & $ \top $ & $ \top $ & $ \top $ & $ \top $ & $ \top $ & $ \top_P $ & $ \top_P $ \\
\midrule
$\countsem(\always \eventually b,i)$ & $(\infty,\infty)$  & $(\infty,\infty)$  & $(\infty,\infty)$  & $(\infty,\infty)$  & $(\infty,\infty)$  & $(\infty,\infty)$  & $(\infty,\infty)$  \\ \nopagebreak 
$\predictive(\always \eventually b,i)$ & $ \top_P $ & $ \top_P $ & $ \top_P $ & $ \top_P $ & $ \top_P $ & $ \top_P $ & $ \top_P $ \\
\midrule
$\countsem(\always \eventually a \lor \always \eventually b,i)$ & $(\infty,\infty)$  & $(\infty,\infty)$  & $(\infty,\infty)$  & $(\infty,\infty)$  & $(\infty,\infty)$  & $(\infty,\infty)$  & $(\infty,\infty)$  \\ \nopagebreak 
$\predictive(\always \eventually a \lor \always \eventually b,i)$ & $ \top_P $ & $ \top_P $ & $ \top_P $ & $ \top_P $ & $ \top_P $ & $ \top_P $ & $ \top_P $ \\
\bottomrule
  \end{tabular}
  }
\end{table}

  
 \begin{table}[tb]
            \centering
  \caption{Trace where evaluations differ for semantically equivalent specifications.}
\label{sec:finiteLTL:sec:Examples:inoutseq:shortcoming}
  {\scriptsize
  \begin{tabular}{ c  c c c c c c c c c c}
  \toprule
   & 1 & 2 & 3 & 4 & 5 & 6  \\ 
 \midrule
 a & $ \top $  & $ - $  & $ \top $  & $ - $  & $ - $  & $ - $ \\ 
b & $ - $  & $ - $  & $ - $  & $ \top $  & $ \top $  & $ \top $ \\ 
 \bottomrule
\end{tabular}
}
 \end{table}

\paragraph{When one subformula dominates:}
We now discuss a shortcoming of our approach. Consider the following specification
$$
\phi= \always(\eventually a \lor \eventually b).
$$
This specification requires that for any index $i$ either signal $a$ evaluates to $\true$ now or at a future position or, otherwise, signal $b$ evaluates to $\true$ now or at a future position. In Table~\ref{sec:finiteLTL:sec:Examples:GFaFb1} we see that our approach concludes the trace under evaluation to presumably false. This is not what we would expect, as for positions smaller than or equal to 4, the formula $\eventually a$ is always satisfied immediately in the same time step and for all observed positions $i\leq 5$ the formula $\eventually b$ is satisfied within in at most one additional time step. In position $i=6$ our approach predicts the formula $\eventually a \lor \eventually b$ to be presumably false, because the shorter witness for satisfaction dominates and, as both of the subformulas are eventually properties, none of them can be violated in finite time. Thus, the globally property is predicted to be violated which results in the evaluation of presumably false.

Intuitively, $\phi$ requires in every time step to eventually raise one of the two signals, i.e., one interpretation is that only the faster satisfaction counts. The specification $\phi'= \always\eventually( a \lor b)$ is semantically equivalent to $\phi$ and expresses this interpretation formally and (also) evaluates to presumably false.

On the other side, if we rewrite $\phi$ to 
$$
\phi'' = \always \eventually a \lor \always \eventually b,
$$ which is again semantically equivalent to $\phi$, then the conclusion is presumably true (see Table~\ref{sec:finiteLTL:sec:Examples:GFaFb2}), which is what we would expect. Thus, there is a difference in the interpretation of $\phi$ (and $\phi'$) and $\phi''$. 
The specification $\phi''$ can be interpreted such that the system only has to satisfy one of the two formulas $\always \eventually a$ and $\always \eventually b$, as those to formulas are connected with a logical or. Thus, the violation of one of the globally properties still allows the specification to be presumably satisfied (by the other globally).

Another example for two specifications that are semantically equivalent, but can be interpreted in different ways is:
\begin{align*}
\psi &= \always(\eventually a \lor \always b)\\
\psi' &=\always\eventually a \lor (\eventually a \until \always b)
\end{align*}

While in specification $\psi$ the formula $\eventually a$ dominates, because the formula $\always b$ cannot be satisfied in finite time, the rewriting to $\psi'$ eliminates this dominating factor. Thus, for the trace presented in Table~\ref{sec:finiteLTL:sec:Examples:inoutseq:shortcoming}, evaluating $\psi$ results in presumably false and evaluating $\psi'$ results in presumably true.


\paragraph{System implements the specification in different modes:}
In the above examples we've shown a weakness of our approach that arises from a dominating subformula.
The specifications with dominating subformulas for which our predictions fail have in common that they implicitly allow systems to operate in two modes and (eventually) switch from one mode to the other. 

Our approach may also fail for a system that operates in different modes when the mode is not part of the specification, e.g., a system that has a high- and a low-performance mode. 
Consider a system that implements the low-performance mode in such a way that the system takes longer to react (without violating the specification).
When the trace contains system behavior of both modes, i.e., the high-performance and the low-performance mode, then our prediction is built on the behavior of the low-performance mode (assuming that witnesses are longer here), as we look at the longest observed witness for satisfaction. Thus, at some point predictions in the high-performance mode may be incorrect.


 \paragraph{Shortcoming of our Approach:}

Consider the specification $\always \eventually p$ and a system that raises $p$ in the time steps $1,2,4, \dots, 2^i$ with $i= 3 \dots \infty$.
As the distance for the next satisfaction of $\eventually p$ always doubles, we will give a wrong evaluation in half of the case.
The reason for the wrong evaluation is that we have not yet observed witnesses with similar lengths for the second half of the last (doubled) distance to the (not yet observed) satisfaction of the eventually part.

\end{document}